\documentclass[sigconf, nonacm]{acmart}

\usepackage{amsmath}
\usepackage{color}
\usepackage{subfigure}
\usepackage{array}
\usepackage{pifont}
\usepackage{{inputenc}}
\usepackage{balance}
\usepackage{multirow,tabularx}
\usepackage{booktabs,longtable}
\usepackage[ruled,vlined]{algorithm2e}
\usepackage{algpseudocode}
\usepackage{graphicx}
\usepackage{bbm}
\usepackage{enumitem}
\usepackage{verbatim}
\usepackage{amsfonts}
\usepackage{hyperref}
\usepackage{enumitem}
\usepackage{color}
\hypersetup{colorlinks=false,linkcolor=blue,urlcolor=blue,citecolor=red}
\usepackage{xspace}
\usepackage{multirow}
\newtheorem{definition}{Definition}

\newcommand\vldbdoi{10.14778/3594512.3594520}
\newcommand\vldbpages{1897 - 1909}
\newcommand\vldbvolume{16}
\newcommand\vldbissue{8}
\newcommand\vldbyear{2023}
\newcommand\vldbauthors{\authors}
\newcommand\vldbtitle{\shorttitle}
\newcommand\vldbavailabilityurl{https://github.com/zealscott/LDPTrace}
\newcommand\vldbpagestyle{empty}

\newcommand{\bigo}{\mathcal{O}}
\newcommand{\lapl}{\mathcal{L}}

\newcommand{\Space}[1]{\mathbb{#1}}
\newcommand{\Set}[1]{\mathcal{#1}}

\newcommand{\ie}{\emph{i.e., }}
\newcommand{\eg}{\emph{e.g., }}

\newcommand{\etc}{\emph{etc.}}
\newcommand{\wrt}{\emph{w.r.t. }}

\newcommand{\aka}{\emph{aka. }}
\newcommand{\argmin}{\mathop{\mathrm{argmin}}}

\newcommand{\mean}{\mathop{\mathrm{E}}}
\newcommand{\var}{\mathop{\mathrm{Var}}}
\newcommand{\error}{\mathop{\mathrm{Error}}}

\newcommand{\mymethod}{\ensuremath{\mathsf{LDPTrace}}\xspace}
\newcommand{\baseline}{\ensuremath{\mathsf{NGRAM}}\xspace}
\newcommand{\randsyn}{\ensuremath{\mathsf{RandSyn}}\xspace}
\newcommand{\combtran}{\ensuremath{\mathsf{CombTran}}\xspace}
\newcommand{\noadapt}{\ensuremath{\mathsf{NoAdapt}}\xspace}

\begin{document}
\title{\mymethod: Locally Differentially Private Trajectory Synthesis}

\author{Yuntao Du}
\affiliation{%
  \institution{Purdue University}
}
\email{ytdu@purdue.edu}

\author{Yujia Hu}
\affiliation{%
  \institution{Zhejiang University}
}
\email{charliehu@zju.edu.cn}


\author{Zhikun Zhang}
\affiliation{%
  \institution{CISPA}
}
\email{zhikun.zhang@cispa.de}

\author{Ziquan Fang}
\author{Lu Chen}
\affiliation{%
  \institution{Zhejiang University}
}
\email{{zqfang,luchen}@zju.edu.cn}


\author{Baihua Zheng}
\affiliation{%
    \institution{Singapore Management University}
}
\email{bhzheng@smu.edu.sg}

\author{Yunjun Gao}
\affiliation{%
  \institution{Zhejiang University}
}
\email{gaoyj@zju.edu.cn}


\begin{abstract}
Trajectory data has the potential to greatly benefit a wide-range of real-world applications, such as tracking the spread of the disease through people's movement patterns and providing personalized location-based services based on travel preference. However, privacy concerns and data protection regulations have limited the extent to which this data is shared and utilized. To overcome this challenge, local differential privacy provides a solution by allowing people to share a perturbed version of their data, ensuring privacy as only the data owners have access to the original information.
Despite its potential, existing point-based perturbation mechanisms are not suitable for real-world scenarios 
due to poor utility, dependence on external knowledge, high computational overhead, and vulnerability to attacks. To address these limitations, we introduce \mymethod, a novel locally differentially private trajectory synthesis framework. Our framework takes into account three crucial patterns inferred from users' trajectories in the local setting, allowing us to synthesize trajectories that closely resemble real ones with minimal computational cost. 
Additionally, we present a new method for
selecting a proper grid granularity without compromising privacy. Our extensive experiments using real-world data, various utility metrics and attacks, demonstrate the efficacy and efficiency of \mymethod.
\end{abstract}

\maketitle

\pagestyle{\vldbpagestyle}
\begingroup\small\noindent\raggedright\textbf{PVLDB Reference Format:}\\
\vldbauthors. \vldbtitle. PVLDB, \vldbvolume(\vldbissue): \vldbpages, \vldbyear.\\
\href{https://doi.org/\vldbdoi}{doi:\vldbdoi}
\endgroup
\begingroup
\renewcommand\thefootnote{}\footnote{\noindent
This work is licensed under the Creative Commons BY-NC-ND 4.0 International License. Visit \url{https://creativecommons.org/licenses/by-nc-nd/4.0/} to view a copy of this license. For any use beyond those covered by this license, obtain permission by emailing \href{mailto:info@vldb.org}{info@vldb.org}. Copyright is held by the owner/author(s). Publication rights licensed to the VLDB Endowment. \\
\raggedright Proceedings of the VLDB Endowment, Vol. \vldbvolume, No. \vldbissue\ %
ISSN 2150-8097. \\
\href{https://doi.org/\vldbdoi}{doi:\vldbdoi} \\
}\addtocounter{footnote}{-1}\endgroup

\ifdefempty{\vldbavailabilityurl}{}{
\vspace{.3cm}
\begingroup\small\noindent\raggedright\textbf{PVLDB Artifact Availability:}\\
The source code, data, and/or other artifacts have been made available at \url{\vldbavailabilityurl}.
\endgroup
}

\section{Introduction}


The widespread availability of location sensing technology, such as GPS, has 
revolutionized our ability to collect real-time data. As a result, there has been significant interest in studying human mobility patterns on a large scale for a variety of location-based applications, including traffic prediction, route planning, and recommendation. Despite the immense value of this data, privacy concerns surrounding the sensitive nature of trajectories have limited its use.

Differential privacy (DP) has become the \textit{de facto} standard for protecting sensitive data while ensuring individual privacy. Despite the development of various DP algorithms for trajectory publishing and analysis, these methods rely on a 
\emph{trustworthy aggregator} to collect users' raw trajectories. In contrast, local differential privacy (LDP) allows users to directly share a noisy version of their data, reducing the risk of data breaches from  
untrustworthy data curators.

The LDP provides a more practical setting and improved privacy properties, however, it imposes challenges in preserving the complex spatial patterns of trajectories due to its strict privacy requirements. 
Currently, the only solution that meets the rigorous privacy requirement of LDP is \baseline~\cite{vldb21ngram}. This method uses the exponential mechanism to directly perturb individual trajectory in the local setting and leverages external knowledge (POIs, business opening hours, \etc) and overlapped n-grams to enhance  the realism of the noisy trajectories.
However, \baseline has several major limitations that hinder its effectiveness: 

\begin{itemize}[leftmargin=*]
    \item \textbf{Poor global utility.} 
    As \baseline only focuses on local trajectory proximity for utility optimization (similarity between original and perturbed trajectories), it results in poor global utilities, as reported in Section~\ref{sec:exp-utility}. Most location-based applications rely on population-level spatial statistics (\eg range query and spatial density) and moving patterns (\eg distribution of start/end points and frequent travel patterns), rather than individual behaviors (\eg routing preference), and the failure to  preserve
    global utility significantly limits its applications.
    \item \textbf{Dependence on auxiliary knowledge.}
    The performance of \baseline heavily relies on external knowledge (\eg POI categories and business opening hours),
    which may not always be accessible and can become outdated easily. This leads to a dramatic decrease in utility and authenticity of generated trajectories. In addition, users are required to store the external data on their own devices, which is highly impractical for wearable or low-cost GPS devices with limited storage.
    \item \textbf{High computational overhead.}
    To obtain accurate perturbed trajectories, 
    \baseline uses linear programming solvers, which are time-consuming and require pre-processing to deal with external knowledge (\eg POI processing and hierarchical decomposition). This can result in significant delays for users, reducing their satisfaction with location-based services.
    \item \textbf{Vulnerable to attacks.}
    Point-based privacy mechanisms are vulnerable to location-based attacks, such as re-identification attack~\cite{nature13,ccs15temporal} and outlier leakage~\cite{ccs18adatrace,icde22gl}, due to the strong statistical correlation between the fake locations and user’s true locations. This is even more concerning for \baseline as it leverages external knowledge to maintain geographic and semantic similarities between real and perturbed points in the local setting (as to be detailed in Section~\ref{sec:exp-attack}). As a result, alternative attack-resilient approaches must be sought to address privacy concerns.
\end{itemize}

The existing shortcomings motivate us to develop a new local privacy-preserving paradigm that is both utility-aware and efficient. Instead of perturbing each trajectory individually,
we aim to extract the key movement patterns of each user and use them to synthesize privacy-preserving and realistic trajectories. However, synthesizing trajectories in the local setting is \emph{challenging} due to the following two reasons.
First, previous trajectory synthesis methods either rely on global statistical metrics~\cite{sp16sglt,ccs18adatrace} or spatial-aware data structures (\eg prefix tree~\cite{kdd12ngram,vldb15dpt}) to model the spatial distributions. However, these are not feasible in  local settings as there is no trusted data curator to collect these information and it is infeasible to directly collect the statistics from individual users, who typically only have a few trajectory footprints, leading to severely biased estimations. 
Second, existing LDP methods assume that each individual holds a single data record (\eg a single value)~\cite{usenix17oue,sigmod18ldp}, but this assumption is no longer valid
for trajectory data, which is a sequence of spatial points with timestamps and has
complex spatial context that cannot be simply modeled by the methods.

Therefore, we present \mymethod,
a simple yet effective framework for synthesizing locally differentially private and attack-resistant trajectories. \mymethod achieves the following four objectives: (i) robust, rigorous statistical privacy, (ii) flexibility and low computational cost, (iii) strong preservation of global spatial utilities and authenticity, and (iv) deterministic resilience against trajectory privacy attacks. Specifically, \mymethod approaches trajectory synthesis as a generative process, constructing a probabilistic model based on users' transition records to estimate the global moving patterns. The transition records capture the spatial relationship between adjacent points, while remaining low computational complexity. To enhance the authenticity of synthesized trajectories, \mymethod includes \textit{virtual} start and end points to each trajectory in the generation process to indicate the beginning and terminated transition states. 
Additionally, the framework estimates the trajectory length distribution for optimal transition budget allocation and deterministic generation constraints. We also provide a theoretical guideline for selecting grid granularity without consuming privacy budget. Last but not the least, an adaptive synthesis algorithm is employed to generate realistic trajectories without access to users' real trajectories.

Our synthetic trajectory generation process is locally differentially private, meaning that the global moving patterns are not strongly dependent on any specific user and the generation of synthetic trajectory is  not 
bias towards any specific trajectory. We have conducted an extensive experimental evaluation to compare \mymethod with the state-of-the-art method \baseline.
Experimental results indicate that \mymethod significantly outperforms 
the competitor in terms of data utility, efficiency, and scalability, and it  
is equipped with robust resistance against various location-based attacks for superior privacy protection.

In summary, the key contributions of our work are below.
\begin{itemize}[leftmargin=*]
    \item We propose \mymethod, the first trajectory synthesis solution with local differential privacy guarantee that is able to generate realistic trajectories without any external knowledge.
    \item We introduce a neat and effective framework that collects key moving patterns from users' trajectories with little computational cost, and devise an adaptive synthesis algorithm to generate authentic trajectories.
    \item We perform comprehensive analysis on the errors and complexity of the proposed framework, and present a guideline for selecting the grid granularity without consuming privacy budget.
    \item We conduct extensive experiments to demonstrate the superiority of \mymethod in terms of utility, efficiency, and scalability. Moreover, we show that \mymethod is able to resist various location-based attacks.
\end{itemize}

The rest of this paper is organized as follows. We review related work in Section~\ref{sec:related}, and elaborate our motivation in Section~\ref{sec:motivation}. Then, we present preliminaries in Section~\ref{sec:preliminaries}. Section~\ref{sec:methodology} details the proposed synthesis framework. The experimental results are reported in Section~\ref{sec:exp}. Finally, we conclude the paper, and offer directions for future work in Section~\ref{sec:conclusion}.

\section{Related Work}\label{sec:related}


Differential privacy (DP)~\cite{dwork_dp} has become the \textit{de facto} privacy standard. While centralized DP assumes data aggregators are reliable, local differential privacy  (LDP)~\cite{13ldp} assumes that aggregators cannot be trusted and relies on data providers to perturb their own data. Early studies~\cite{vldb12privbasis,ccs13privacy,nips17hitter,icde17privsuper,vldb18,gu2019supporting,cormode2021frequency,sigmod20ldp} on DP and LDP mostly focus on designing tailored algorithms for specific data analysis tasks, which suffers from poor flexibility, inefficiency, and scalability problems. One promising solution~\cite{usenix21privsyn,vldb21markov,vldb21kamnio,ccs18clam,YZDCCS23} to address this problem is generating a synthetic dataset that is similar to the private dataset while satisfying (local) differential privacy.
However, these methods mainly aim at structured data like tables, which cannot be applied to trajectory data due to its high dimensionality and complex spatial dependence.


The privacy of trajectory data (surveyed in~\cite{survey21,tkde22survey}) has been a significant concern for over a decade and various solutions have been proposed to address the issue.
Many existing solutions~\cite{cikm13,ccs2013geo,icde22gl,vldb22query,ccs21dp,kdd14spatio,icde17temporal} employ spatial point perturbation techniques under the constraint of DP
that add noises to the point locations of trajectory before it is published or used to answer predefined queries (\eg range query). For example, GL~\cite{icde22gl} aims to preserve both privacy and high utility by perturbing the local/global frequency distributions of important locations in a trajectory. 
SNH~\cite{vldb22query} introduces a neural database for spatial range queries and adds DP-compliant noise to the input queries to maintain the density features of location data. More recently, \baseline~\cite{vldb21ngram} has been proposed to address the privacy concerns of trajectory sharing in a local setting. 
The method includes three phases: (1) hierarchical decomposition phase, where  POIs are divided into spatial-temporal-category regions; (2) perturbation phase, where trajectories are converted to sequences of overlapping n-grams and perturbed using exponential mechanism; and (3) reconstruction phase, where \baseline solves an optimization problem to reconstruct the continuous trajectory based on the perturbed n-grams. Despite its efforts, \baseline still has several limitations, such as low global utility, high computational overhead, and vulnerability to attacks.

To protect privacy, some researchers have investigated alternative methods to point perturbation, such as generating synthetic trajectories~\cite{kdd12ngram,vldb15dpt,sp16sglt,mc18,ccs18adatrace,sstd21,gursoy2020utility,yang2022collecting}. The challenge is to create synthetic data that resembles genuine user traces while providing practical privacy protection.
One approach, DPT~\cite{vldb15dpt}, uses a hierarchical reference system to model trajectory movements at different speeds and encodes transitions between grid cells using prefix trees. By injecting Laplace noise to the prefix trees, the transition probabilities are distorted while still maintaining the movement patterns of the original traces. Further studies~\cite{sp16sglt,mc18} have extended DPT by incorporating trajectory semantics and temporal information. 
Another method, AdaTrace~\cite{ccs18adatrace}, extracts four statistical and spatial features and incorporates them into its private synopsis, including a density-aware grid, Markov mobility model, pickup/destination and length distribution. The authors also design a synthesizer with attack resilience constraints to balance both statistical privacy (differential privacy) and syntactic privacy (attack resilience).
PrivTrace~\cite{usenix2023privtrace} applies an adaptive strategy to choose crucial first- and second-order Markov transitions for trajectory synthesis, resulting in better utilities than previous methods. Although current synthetic methods are effective, they all rely on a trustworthy data curator to aggregate useful statistics. Our work is the first to introduce an utility-aware and efficient trajectory synthesis framework without accessing users' real traces.

\section{Motivation}\label{sec:motivation}

Trajectory is a time-order sequence of location points generated from human mobility, which is highly sensitive since it can reveal people's home/work place, travel patterns, and other preferences.
The population level aggregate spatial information on where/when do residents commute could be useful for the authority to gain better understanding of residents' commuting patterns,
but many people are unwilling to share their own trajectories because of privacy concerns. Therefore, we aim to propose a method for various parties (e.g., authority/service providers) to collect useful mobility patterns from crowd without accessing individual's real trajectories. In the following, we present four critical design principles that motivate and guide our solution, including \emph{privacy protection}, \emph{global utility}, \emph{practicability and efficiency}, and \emph{attack resilience}.

\vspace{2pt}
\noindent
\textbf{Privacy protection.} The primary goal of our work is
to protect each individual's privacy so that the untrusted data curator cannot access people's real traces. We achieve this by utilizing LDP mechanism to perturb user's trajectories
before sharing the data. We detail the privacy implications of our method in Section~\ref{sec:privacy-analysis}.

\vspace{2pt}
\noindent
\textbf{Global utility.} Based on the rigorous privacy guarantee of LDP, 
our solution is expected to be designed in such a way that it can preserve the high global utility for synthetic trajectories. We argue that a feasible way to boost global utility is to extract the key moving patterns from user's trajectories, and leverage these information to guide the synthesis process. Since the synthesizer is designed to capture the intrinsic features of user's movements, the synthetic trajectories can be used for various spatial analysis tasks like range query and frequent pattern mining, instead of 
tailor-made for specific utility like~\cite{vldb21ngram}.

\vspace{2pt}
\noindent
\textbf{Practicability and efficiency.} Considering the 
wide usage of trajectory data, our solution is preferred to be as simple as possible so that it can be easily deployed to any local devices without heavy computation. Besides, it is equally critical to ensure the efficiency, thus privacy and utility do not come at the cost of user experience.

\vspace{2pt}
\noindent
\textbf{Attack resilience.} Although LDP is a strong data sharing technique with provable privacy guarantees, it still suffers from various syntactic attacks. This is especially true for trajectory privacy protection, as blindly forcing the perturbed trajectories to resemble original ones would make them vulnerable to various location-based attacks like re-identification attack and outlier leakage. Hence, the synthetic trajectories should be robust, and they are expected to be able to resist common attacks.

\vspace{2pt}
\noindent
\textbf{Applications.} Our work focuses on synthesizing trajectories such that the global aggregate statistics is preserved as much as possible, 
which is essential to many important location-based applications. A notable one is 
trajectory monitoring that identifies people's movement patterns, which can be used for policy making decisions like traffic control or disease spread forecasting (like Covid-19). Other applications include location-based services and advertising, \eg a tourism recommendation system can utilize common trajectories taken by people to recommend popular trips/destinations, and an outdoor advertising company can use people's movement patterns to better estimate the traffic flows at different locations.



\begin{table}[t]
    \centering
    \caption{Symbols and Description}
    \vspace{-10pt}
    \resizebox{0.47\textwidth}{!}{
    \begin{tabular}{cl}
    \hline
         \textbf{Notation} & \textbf{Description} \\ \hline
         $T$, $\Set{T}$  & Trajectory, and a set of trajectories \\
         $C$, $\Set{C}$  & Grid cell, and a set of grid cells \\
         $L$, $\Set{L}$  & Trajectory length, length distribution \\
         $s$, $\Set{S}$  & Intra-trajectory transition, mobility model \\
         $M$, $\Set{M}$  & Aggregated transition, aggregated mobility model \\
         $C_a$, $C_b$  & Virtual start/end point \\
         $\Set{N}$, $\Set{N}^*$ & Neighborhood cells without/with virtual end point \\
         $\hat{g}(\cdot)$, $\tilde g(\cdot)$ & Report times, unbiased estimation of frequency \\
         $\epsilon$ & Privacy budget  \\  \hline
    \end{tabular}
    }
    \vspace{-5pt}
    \label{tab:notation}
\end{table}

\section{Preliminaries}\label{sec:preliminaries}

In this section, we first introduce the definition of
local differential privacy (LDP) and then formulate our problem. Table~\ref{tab:notation} lists the notations used in this paper.

\subsection{Differential Privacy in the Local Setting}
\subsubsection{\textbf{$\epsilon$-Local Differential Privacy.}}
In the local setting of DP, there are many \textit{users} and one untrusted \textit{aggregator}, and each user perturbs the input value $x$ using an algorithm $\Psi$ and sends $\Psi(x)$ to the aggregator. The formal privacy requirement is that the algorithm $\Psi(\cdot)$ satisfies the following property:
\begin{definition}[$\epsilon$-Local Differential Privacy]
\label{def:ldp}
    An algorithm $\Psi(\cdot)$ satisfies $\epsilon$-local differential privacy ($\epsilon$-LDP), where $\epsilon \ge 0 $, if and only if for any input $x_1$, $x_2$ and output $y$:
    {\setlength{\abovedisplayskip}{3pt}
    \setlength{\belowdisplayskip}{3pt}
    \begin{equation}
        Pr[\Psi(x_1)=y]\leq e^\epsilon Pr[\Psi(x_2)=y].
    \end{equation}}
\end{definition}

The privacy budget $\epsilon$ is a metric used to measure the level of privacy protection in local differential privacy (LDP).
It represents the probability that an attacker can determine the true value of the input based on the output. The higher the privacy budget, the higher the probability that the attacker can infer the true value, and the lower the privacy protection.
In practice, the privacy budget can be set according to the privacy  requirements of the application.
For example, a smaller privacy budget may be chosen when collecting highly sensitive data such as health information, while a larger privacy budget may be used for less sensitive data such as typing patterns. A privacy budget of less than 2 is typically considered acceptable~\cite{ccs14rappor,icde22gl,nips17microsoft,usenix17oue,vldb22query,vldb21ngram,WCZSCLLJ21,DZBLJCC21}.
LDP provides privacy protection by allowing the user to report a perturbed version of the input $\Psi(x)$ instead of the true value $x$ to the aggregator. This ensures that even if the aggregator is malicious, the user's privacy is still protected.
LDP possesses two fundamental properties used in our mechanism~\cite{sigmod18ldp}: the composition theorem, which states that $k$ $\epsilon_i$-LDP mechanisms can be combined to achieve $\epsilon$-LDP protection, where $\epsilon=\sum_i \epsilon_i$;
and the ability to perform post-processing on private
outputs without affecting the privacy guarantee. 
For more information on LDP, please refer to  recent surveys~\cite{sigmod18ldp,survey20ldp,arxiv20ldp}.

\subsubsection{\textbf{Optimized Unary Encoding.}}
A \textit{frequency oracle} (FO) protocol enables the estimation of the frequency of any value $x$, which serves as a building
block of many LDP tasks. In this paper, we opt for Optimized Unary Encoding (OUE) as the FO protocol to achieve frequency estimation under LDP, which consists of three stages: encoding, perturbing, and aggregation~\cite{usenix17oue}.
\begin{itemize}[leftmargin=*]
    \item \textbf{Encoding.} The original value $x$ is first encoded as a length-$d$ binary vector $V$, where only the $x$-th bit is set to 1, \ie $V[i]= \mathbbm{1} (i==x)$.
    \item \textbf{Perturbing.} When reporting the encoded vector $V$, it is perturbed as follows:
    {\setlength{\abovedisplayskip}{3pt}
    \setlength{\belowdisplayskip}{3pt}
    \begin{equation}\label{equ:OUE-perturb}
        Pr[\hat{V}[i]=1] =
        \begin{cases}
            \frac{1}{2}, & \text{if } V[i] = 1\\
            \frac{1}{e^\epsilon + 1}, &\text{if } V[i] = 0,
        \end{cases}
    \end{equation}}where $\epsilon$ is the privacy budget and $\hat{V}$ is the reported noise vector.
    \item\textbf{Aggregation.} In order to obtain the unbiased estimation of the real value from noise vectors, the data curator needs to aggregate and adjust the received data as follows:
    {\setlength{\abovedisplayskip}{3pt}
    \setlength{\belowdisplayskip}{3pt}
    \begin{equation}\label{equ:OUE-unbiased}
    \tilde g(x) = \frac{\hat{g}(x)-nq}{\frac{1}{2}-q},\ q=\frac{1}{e^\epsilon + 1},
    \end{equation}}where $n$ is the total number of reported noise vectors, 
    and $\hat{g}(x)$ is the total number of the reported vectors $\hat{V}$ whose $x$-th bit is 1, \ie $\hat{g}(x)=|\{\hat{V}|\hat{V}[x]=1\}|$. Notice that this adjustment requires the budget $\epsilon$ to remain the same across all the reported data.
\end{itemize}
It can be theoretically proved~\cite{usenix17oue} that the adjusted estimation $\tilde{g}(x)$ is unbiased. The mean and variance of OUE are listed below.
{\setlength{\abovedisplayskip}{3pt}
\setlength{\belowdisplayskip}{3pt}
\begin{equation}\label{equ:oue-variance}
    \mean[\tilde{g}(x)] = f_x, \quad \var[\tilde{g}(x)]=n\frac{4e^\epsilon}{(e^\epsilon-1)^2},
\end{equation}}where $f_x$ is the frequency of value $x$, and $\epsilon$ is the privacy budget.

\subsection{Problem Statement}

Consider there is a collection of trajectories generated by mobile travelers on the roads, denoted by $\Set{T}$. People are unwilling to report their own trajectories to the untrusted data curator due to privacy concerns. Thus, we want to build a generative model over $\Set{T}$ by extracting key moving patterns from $\Set{T}$
with provable guarantees on individual privacy. We then employ the generative model to output a set of synthesized trajectories, denoted by $\Set{T}_{syn}$. The synthesized trajectories $\Set{T}_{syn}$ should collectively retain a high resemblance to the real trajectories $\Set{T}$, so that $\Set{T}_{syn}$ has many useful statistical and spatial features in common with $\Set{T}$. Finally, 
the synthetic trajectories $\Set{T}_{syn}$ should be robust against various location-based attacks, in order to strengthen privacy by maximizing attackers' probability of errors in identifying the true 
traces of users.


\section{Our Methodology}\label{sec:methodology}

In this section, we first give an overview of our proposed method \mymethod, and discuss the necessity and method for trajectory discretization. Then, we detail the three crucial components of \mymethod, which are estimated from collaborative mobility patterns and satisfy the strict privacy guarantee of LDP. Next, we detail the adaptive generation process by building a probabilistic model with learned spatial patterns. Finally, we conduct thorough discussion to analyze the privacy and computational cost of \mymethod, and present an approach to select a proper grid granularity.

\begin{figure}
    \centering
    \includegraphics[width=0.49\textwidth]{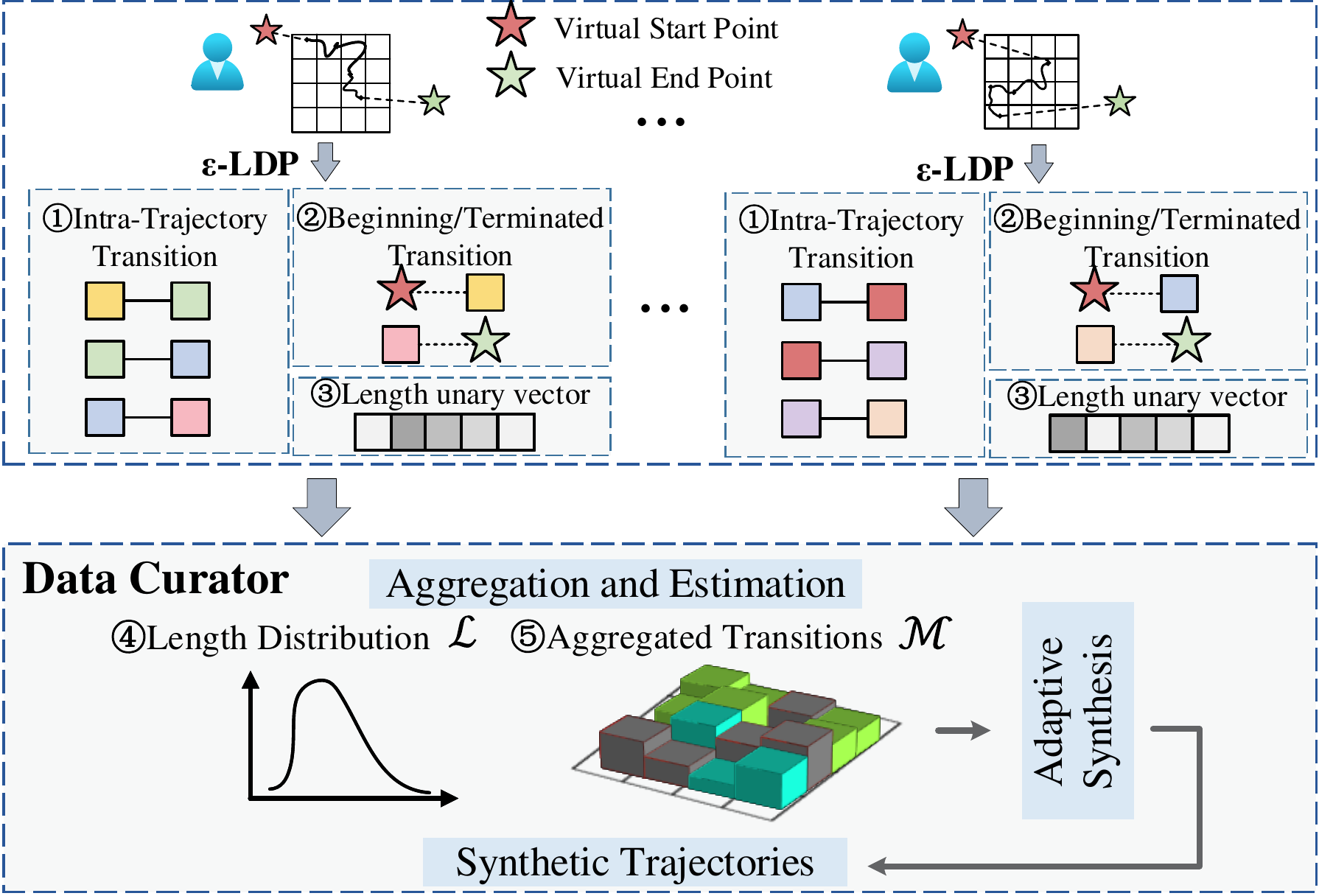}
    \vspace{-20pt}
    \caption{The \mymethod framework.}
    \label{fig:framework}
    \vspace{-10pt}
\end{figure}

\subsection{Solution Overview}
As illustrated in Figure~\ref{fig:framework}, both user and data curator participate in the process of \mymethod. On the user side, \mymethod first discretizes the trajectory into a sequence of 
adjacent cells, and exacts three key features of the trajectory by perturbing the trajectory length, intra-trajectory mobility, and beginning/terminated transitions. On the data curator side, \mymethod collects the perturbed information (\ie three features mentioned earlier) from all users, and builds a probabilistic model by estimating the key mobility features (\ie length distribution and aggregated transitions). Finally, \mymethod performs the synthesis 
to sample start point and subsequent 
transitions from the learned patterns, with its termination controlled 
by the estimated trajectory length and sampled end point.

Since the synopsis of \mymethod consists of three features learned from user's trajectory, the security budget $\epsilon$ is divided into three sub-budgets $\epsilon_1$, $\epsilon_2$, and $\epsilon_3$, such that $\sum\nolimits_{i=1}^3 \epsilon_i = \epsilon$. Detailed analysis on the privacy is presented in Section~\ref{sec:privacy-analysis}.



\subsection{Geospatial Discretization}

The representation of a trajectory as a sequence of points in a continuous two-dimensional domain, such as latitude-longitude coordinates, can be challenging to model. To overcome this, one common approach is to discretize the geographic space into grid cells. This is achieved by partitioning the entire space into equal-sized cells using a grid granularity of $N$.

Each trajectory is then transformed into a sequence of enumerable cells, $T = \{C_1, C_2,\cdots, C_{|T|}\}$, where $T[i]$ refers to the $i$-th cell visited by the trajectory $T$ and $|T|$ is the length of the trajectory in grid cells. The choice of $N$ affects the size of each grid cell, and thus the granularity of the discretized trajectory.


When $N$ is small, the space is partitioned into a limited number of grid cells, and each cell covers a large spatial region. Many points in the original trajectory will be represented by one single cell, and thus, the mobility patterns captured by the discretized trajectory would be very general and uninformative. On the other hand, if $N$ is big, each grid cell covers a very small region, resulting in the risk of having many empty cells that have not been passed by any trajectories. Perturbing these empty cells leads to high noise and inefficiency.
%
While some literature~\cite{icde13grid,ccs18adatrace} has provided  guidelines for selecting a proper grid granularity, they all rely on global statistics like spatial density, which are typically not available in the local setting where a user only has access to his/her own trajectories.
Instead of estimating the global statistics which costs extra privacy budget, we theoretically analyze the trade-off between estimation errors and the granularity of grid cells,
and propose a novel method to choose
$N$ in the local setting without consuming any budget, as to be detailed in Section~\ref{sec:guidance-of-grid}.

\subsection{Trajectory Length Distribution}\label{sec:lengh-distribution}

The first component in \mymethod is the estimated length distribution of trajectories, which is an indispensable character of trajectory synthesis since it indicates the probabilistic travel distance of users' trace, and serves as the deterministic constraint to terminate the synthesis process.
However, estimating the length distribution is a non-trivial task in the local setting since there is no aggregated statistics available. Instead, we aim to collect the length information from each user and approximate the distribution with frequency.
Specifically, we first define the domain of lengths as $|\Set{C}|$, under the assumption that the maximum travel distance of any trajectory is $|\Set{C}|$. Note that the techniques proposed in this paper can take any maximum distance as an input.  

On the user side (\ding{174} in Figure~\ref{fig:framework}), for trajectory $T$ with length $m$ (\ie $|T|=m$), we encode it into a $|\Set{C}|$-bit binary vector $V$, which sets only the $m$-th bit to 1 but all other bits to zero.
Next, the binary vector $V$ is perturbed individually and locally with budget $\epsilon_1$ according to Equation (\ref{equ:OUE-perturb}), and the user only reports the noisy vector $\hat{V}$ to the untrusted data curator.

On the data curator side (\ding{175} in Figure~\ref{fig:framework}), after collecting noisy vector $\hat{V}$ from different users, the curator estimates the frequency of each length value $m$ by counting the non-zero entry of each vector and adjusting the total count of each length via the unbiased statistic using Equation (\ref{equ:OUE-unbiased}). Finally, we view the length distribution as the categorical distribution $\lapl$, and the probability of length $m$ is $Pr(m) = \tilde{g}(m)/\sum\nolimits_{i=1}^{|\Set{C}|} \tilde{g}(i)$, where $\tilde{g}(m)$ is the unbiased OUE estimator of the true frequency of length $m$.

\vspace{3pt}
\noindent
\textbf{Error analysis.} We analyze the error of estimated length distribution $\lapl$ by using the Theorem~\ref{theorem-ratio} in Appendix~\ref{appendix-proofs}. $Pr(m)$ is the unbiased estimator of length probabilities: 
{\setlength{\abovedisplayskip}{3pt}
\setlength{\belowdisplayskip}{3pt}
\begin{equation}
{\mean}^*[Pr(m)]  = f_m / \sum\nolimits_{i=1}^{|\Set{\Set{C}}|} f_i,
\end{equation}}where $f_i$ is the frequency of trajectory length $i$ in the whole trajectory set $\Set{T}$.
Thus, the estimated length distribution $\lapl$ approximates the true distribution with the error:
{\setlength{\abovedisplayskip}{3pt}
\setlength{\belowdisplayskip}{3pt}
\begin{equation}
\begin{aligned}
\error(\lapl) = \sum\nolimits_{i=1}^{|\Set{C}|} (\frac{f_i}{|\Set{T}|})^2 [\frac{\sigma^2}{(f_i)^2} - \frac{2\sigma^2}{f_i |\Set{T}|} + \frac{8\sigma^2}{|\Set{T}|^2}],
\end{aligned}
\end{equation}}where $\sigma^2 = |\Set{T}|\frac{4e^{\epsilon_1}}{(e^{\epsilon_1}-1)^2}$ is the variance of OUE estimator in Equation (\ref{equ:oue-variance}), and $|\Set{T}|$ is the number of trajectories. Given the fixed statistics of trajectories (e.g., size and frequency of each length), the error of estimated length distribution $\lapl$ is on the order of $e^{-\epsilon}$, \ie a higher budget could help reduce the length error.




\subsection{Intra-Trajectory Mobility Model}\label{sec:transition-model}

To generate high utility and realistic synthetic trajectories, it is necessary to mimic actual intra-trajectory mobility (\ie the transition from $T[i]$ to $T[i+1]$). We achieve this by building a Markov chain for mobility modelling. A first-order Markov chain asserts that a location $T[l+1]$ in a trajectory depends only on its previous location $T[l]$ instead of all previous locations:
{\setlength{\abovedisplayskip}{3pt}
\setlength{\belowdisplayskip}{3pt}
\begin{equation}
Pr(T[l+1]=C~|~T[1]...T[l]) = Pr(T[l+1]=C~|~T[l]),
\end{equation}}which simplifies the complex sequential dependency $T[1]...T[l]$ with closest grid $T[l]$ for $T[l+1]$. We assume that the main mobility patterns of trajectories can be captured by the Markov chain, which is a collection of such probability $Pr(T[l+1]=C_{next}|T[l])$.
Considering the 
continuity of trajectories, we interpolate each trajectory to make sure that any two consecutive points in a trajectory corresponding to two adjacent grid cells, \ie $T[i]$ is adjacent to $T[i+1]$.
Accordingly, we can define the transition state $s_{ij}$ from grid cell $C_i$ to another grid cell $C_j$ as follows:
{\setlength{\abovedisplayskip}{3pt}
\setlength{\belowdisplayskip}{3pt}
\begin{equation}
Pr(s_{ij}) =
        \begin{cases}
            Pr(T[l+1]=C_j~|~T[l] = C_i), & \text{if } C_j \in \Set{N}_{C_i} \\
            0, &\text{otherwise},
        \end{cases}
\end{equation}}where $Pr(s_{ij})$ is the transition probability from $C_i$ to $C_j$, $\Set{N}_{C_i}$ captures the set of adjacent cells of $C_i$, and the transition model $\Set{S}$ is a set of all 
transition probabilities aggregated from trajectories.

Similar to the estimation of length distribution, we aim to collect the transition information from each user (\ding{172} in Figure~\ref{fig:framework}) and then aggregate them on the data curator side to estimate the overall transition states. Thus, each user's  trajectory $T$ can be represented as a sequence of transition states $S_T$ with length $|T|-1$. On the user side, for each state in $S_T$, we also opt for OUE to encode it into a $|\Set{S}|$-bit binary vector and then report the perturbed noisy version, where $|\Set{S}|$ is the domain of possible states. Since we only consider transition between two adjacent grids and each grid cell has up to 8 adjacent cells, we have $|\Set{S}|\approx 8|\Set{C}|$.

However, if we directly perturb each trajectory's transition by averaging the privacy budget $\epsilon_2$ with its own length $|S_T|$, it is impossible to estimate the unbiased frequency of each transition state $s$ from all trajectories since OUE protocol demands the budget used to be same across different users/trajectories. A straightforward solution to this problem is to equally divide budget $\epsilon_2$ across $|\Set{C}|$, the maximum length of $T$,
meaning that each trajectory is allowed to have up to $|\Set{C}|$ transitions,  
and each transition is assigned a budget of $\epsilon_2/|\Set{C}|$. Nevertheless, this approach suffers from huge waste of budget as the length of most trajectories could be far shorter than $|\Set{C}|$. Therefore, we set the number
of transitions as the $k$-th quantile of estimated length distribution $L_k$, so that we only upload $L_k$ transition states for all trajectories, and omit the remaining if the number of transitions $|S_T|$ is longer than $L_k$. Here, $k$ is a hyperparameter to balance the noise error and bias error. On one hand, we would like $k$ to be large, \ie every transition in trajectory data $\Set{T}$ can be captured for modelling mobility patterns, and the amount of noise added to state $s$ can be measured by:
{\setlength{\abovedisplayskip}{3pt}
\setlength{\belowdisplayskip}{3pt}
\begin{equation}\label{equ-noise-transition}
N(s,\epsilon_2, L_k) = {\var}^*[Pr(s), \epsilon_2/L_k],
\end{equation}}where $\var^*(Pr(s), \epsilon_2/L_k)$ is the approximated variance of transition probability $Pr(s)$ with budget $\epsilon_2/L_k$, which can be calculated via Theorem~\ref{theorem-ratio}. On the other hand, omitting the transitions for trajectories with length longer than $L_k$
will introduce bias into the model (\ie some of the transitions
are not accurately captured), and the bias of transition $s$ is expressed as $(1-k)^2\cdot f_s^2$,
where $f_s$ is the frequency of transition $s$. 
The sum of the noise and bias terms
defines the total error of $\Set{S}$:
{\setlength{\abovedisplayskip}{3pt}
\setlength{\belowdisplayskip}{3pt}
\begin{equation}\label{equ-error-transition}
\error(\Set{S}, \epsilon_2, L_k) = \sum\nolimits_{s\in \Set{S}} [ N(s,\epsilon_2, L_k) + (1-k)^2\cdot f_s^2 ].
\end{equation}}If $k$ is large, more transition information would be reported, which can effectively reduce the bias error. 
However, the budget of each transition will be small, which results in large noise for estimated frequency and hence more noise error. On the other hand, if $k$ is small, many useful transition states are omitted, leading to insufficiency to capture the global moving patterns from crowds. Therefore, the optimal $k$ is chosen to optimize the total error of $\Set{S}$:
{\setlength{\abovedisplayskip}{3pt}
\setlength{\belowdisplayskip}{3pt}
\begin{equation}\label{equ-optimal-k}
k^* = \argmin\nolimits_{ 0 <  k  \le 1 }~\error (\Set{S}, \epsilon_2, L_k).
\end{equation}}However, due to the unavailability of true frequency of transitions in the local setting, it is impossible to directly derive the optimal $k$. We analyze the impact of $k$ on utility in Section~\ref{sec:exp-impact-of-k}.



\vspace{3pt}
\noindent
\textbf{Discussion.} In \mymethod, we model the intra-trajectory moving patterns with the first-order Markov chain (\aka transition states), while we can naturally apply a higher-order Markov chain. 
However, since the size of Markov chain grows quickly \wrt the order, the noise introduced by the randomness mechanism in LDP can drown out the real signal due to the limited budget. Empirical results in Section~\ref{sec:exp-utility} show that the first-order Markov chain is sufficient to model the intra-mobility, and generates authentic trajectories.

\subsection{Beginning/Terminated Transitions}\label{sec: trip-distribution}

Real-world trajectory databases often consist of various trips, such as taxi trips, home-work commutes, \etc~These trips usually exhibit special start point and end point, which reveal the important spatial semantic of trajectories: pickups, home/work places, destinations, \etc~ Besides, the start/end points are also useful to guide the random walk during synthesis. Although we have a mobility model $\Set{S}$ for intra-trajectory movement, we still need a start state and an end state to specify the two endpoints of generated trajectory. 
We can naively assume a uniform distribution to randomly choose a cell from $C$ as a start point to perform random walk and to terminate the synthesis process when the trajectory reaches the assigned length, but it contradicts with a well-known fact that the distribution of start/end points of real trajectories is often heavily skewed~\cite{ccs18adatrace}. For example, many trips might start from or end at homes, while residential areas are not uniformly distributed in a city. Thus, naive solution mentioned above, even though simple, will incur bogus synthetic trajectories that jeopardize utility and authenticity.

To model the distribution of start/end points in trajectories, we add two special cells, namely, \textit{virtual} start point $C_a$ and \textit{virtual} end point $C_b$, which are connected to all the geographic cells in $\Set{C} $.
$C_a$/$C_b$ serves as
start/end point of any trajectories $T$ to record $T$'s beginning/terminated state.
To be more specific, the beginning transition state $A_i$ denotes that trajectory $T$ begins with $C_i$ (from virtual point $C_a$ to cell $C_i$), and the terminated transition state $B_j$ means that trajectory $T$ stops with $C_j$ (from $C_j$ to virtual end point $C_b$).
On the user side, we also utilize the OUE protocol to report a noisy version of the beginning/terminated transition states with budget $\epsilon_3$ on the user side, which is separated from the intra-trajectory transitions during reporting to reduce their noise since they only occur once per trajectory (\ding{173} in Figure~\ref{fig:framework}). On the data curator side, we combine the intra-transitions with beginning/terminated transitions to form the aggregated mobility model $\Set{M}$ (\ding{176} in Figure~\ref{fig:framework}). Specifically, the transition probability from $C_i$ to $C_j$ can be calculated as:
{\setlength{\abovedisplayskip}{3pt}
\setlength{\belowdisplayskip}{3pt}
\begin{equation}
Pr(M_{ij}) = \frac{\tilde{g}(M_{ij})}{\sum\nolimits_{r\in \Set{N}^*_{C_i}} \tilde{g}(M_{ir})},
\end{equation}}where $M_{ij}$ represents the transition state from $C_i$ to $C_j$. If $C_i$ is the virtual start point $C_a$, then $M_{ij} = A_j$; if $C_j$ is the virtual end point $C_b$, then $M_{ij} = B_i$; otherwise $M_{ij} = s_{ij}$. The aggregated neighbor $\Set{N}^*$ is defined as follows:
{\setlength{\abovedisplayskip}{3pt}
\setlength{\belowdisplayskip}{3pt}
\begin{equation}\label{equ:agg-neighbor}
        \Set{N}^*_{C_i} =
        \begin{cases}
            \Set{N}_{C_i}\cup \{C_b\}, & \text{if } C_i \in \Set{C}\\
            \Set{C}, &\text{otherwise } (\ie C_i \in \{C_a, C_b\}).
        \end{cases}
\end{equation}}Hence, the intra-transitions and the  beginning/terminated transitions are seamlessly integrated into
the aggregated mobility model $\Set{M}$, which can denote the overall moving patterns of trajectories.

\subsection{Trajectory Synthesis}

\mymethod builds a  probabilistic model for private synopsis, which consists of length distribution $\lapl$ and  aggregated mobility model $\Set{M}$.
%
Accordingly, the synthesis algorithm can be described in three steps, as depicted in Algorithm~\ref{alg:generation}. First, it determines the length $L$ of trajectory by sampling from the length distribution $\lapl$ (line 1). Second, it initializes $T_{syn}$ by assigning its starting point to the cell
sampled from $\Set{N}_{C_a}^*$ with probability proportional to $\Set{M}$ (lines 2-3). Third, it extends $T_{syn}$ by including a new cell $C_{next}$ based on its current location. It repeats the extension process until $C_{next}$ is the virtual end point $C_b$ or the length of $T_{syn}$ reaches $L$, whichever is earlier (lines 4-10).
%
Although the above synthesis scheme makes full use of the estimated patterns, a notable weakness is that the generated trajectory $T_{syn}$ might be much shorter than $L$ if it reaches $C_b$ in the early part of its extension.
%
Since the distribution of trajectory terminated points is heavily skewed (\ie some points are more likely to be the destination than others), directly sampling the end point according to the transition probability would lead to the inauthenticity and uninformativeness of synthetic trajectories. Thus, we re-weight the terminated probability $Pr(M_{ib})$ by taking the current length $l$ into consideration:
{\setlength{\abovedisplayskip}{3pt}
\setlength{\belowdisplayskip}{3pt}
\begin{equation}\label{equ:reweight}
\tilde{Pr}(M_{ib}) = (\alpha + \beta l) \times Pr(M_{ib}),
\end{equation}}where $\tilde{Pr}(M_{ib})$
is the adjusted termination probability from  current location $C_i$ to virtual end point $C_b$, and $\alpha$ and $\beta$ are two hyperparameters to control the influence of length, \ie large $\alpha$ and $\beta$ implies that the synthesis process tends to be stopped even when the length is small, and small $\alpha$ and $\beta$ allows the model to synthesize trajectories with more points. The final adaptive synthesis algorithm is presented in Algorithm~\ref{alg:generation}.

To obtain an authentic trajectory database with good utilities, we run the synthesis algorithm multiple times until the number of synthetic trajectories reaches the number of trajectories in $\Set{T}$.
Then, this set of trajectories $\Set{T}_{syn}=\cup T_{syn}$ becomes the substitution of the real trajectory set $\Set{T}$ for various spatial analysis tasks without sacrificing users' privacy.





\begin{algorithm}[t]
\small
\caption{Trajectory synthesis}
\label{alg:generation}
\LinesNumbered
\KwIn{a grid $\Set{C}$, a length distribution $\lapl$, an aggregated mobility model $\Set{M}$, virtual start point/end point $C_a$/$C_b$}
\KwOut{a candidate synthetic trajectory $T_\text{syn}$}

trajectory length $L \gets$ sample($\lapl$)\\
sample  $C_{start}\gets$ from $\Set{N}^*_{C_a}$ with probability proportional to $\Set{M}$\\
initialize $T_{syn}$: $T_{syn}[1]\gets C_{start}$\\
\For{$ l \gets 2$ to $L$}{
    reweight the terminated condition according to Equation (\ref{equ:reweight}) \\
    sample $C_{next}$ from $\Set{N}^*_{T_{syn}[l-1]}$ with probability proportional to $\Set{M}$ \\
    \eIf {$C_{next} = C_b$}
    {
        \Return $T_{syn}$ \\
    }{
        set $T_{syn}[l] \gets C_{next}$
    }
}
\Return $T_{syn}$
\end{algorithm}

\subsection{Privacy Analysis}\label{sec:privacy-analysis}

We now analyze the privacy of the synthesis solution through a sketch proof, and discuss the budget allocation strategy.

\begin{theorem}\label{theorem:ldp}
The while process of \mymethod satisfies $\epsilon$-LDP. 
\end{theorem}
\begin{proof}[Proof of Theorem~\ref{theorem:ldp}]
\mymethod treats $\epsilon$ as the total privacy budget, and distributes it to three sub-budgets (one for each key feature in the synopsis) such that $\sum_{i=1}^3 \epsilon_i = \epsilon$. Perturbing and reporting a feature (length, intra-transition, and start/end transition) consumes the $\epsilon_i$ allocated to it, and thus, depleting the total $\epsilon$ after the reporting phase is complete, based on the sequential composition property. Besides, during the intra-mobility modelling, we divide the budget $\epsilon_2$ for each transition state of trajectory, which also satisfies the sequential composition property. Then, any additional data-independent operations like aggregation and sampling on the perturbed statistics are viewed as post-processing. As a result, \mymethod remains $\epsilon$-locally differentially private.
\end{proof}
\vspace{-1pt}
The budget allocation of $\epsilon$ can be configured by \mymethod automatically or according to the demands of real applications.
The current implementation of \mymethod comes with a default budget distribution, which was empirically determined to yield high average utility: $\epsilon_1 = \epsilon/10$ for length distribution, $\epsilon_2+\epsilon_3 = 9\epsilon/10$ for transitions, where we allocate equal budget to each transition (\ie intra-transitions and beginning/terminated transitions). We find that the estimation of length distribution only requires a small budget, while the transitions of trajectories consume the major budget since they determine the synthesis process. Although the above strategy is beneficial from the overall utility maximization perspective, the budget allocation is very flexible, and it can be easily assigned based on the specific needs of various applications (\eg pick-up identification may need more budget for beginning/terminated transition estimation).

\subsection{Computational Cost}\label{sec:computationl-cost}

We also discuss the computational cost of the proposed \mymethod, which highlights how it is far more practical than other alternatives. We analyze the computational cost from both user side and data curator side in the following.

\vspace{3pt}
\noindent
\textbf{Computation on users.} We assume that each user keeps one trajectory of his/her own. Since each perturbation is done locally on individual device, we only analyze the computational cost for each user/trajectory. First, the cost of perturbing a binary vector of length is $\bigo(|\Set{C}|)$, where $|\Set{C}|$ is the number of grids. Second, the computational complexity of intra-transition modelling is $\bigo(|T||\Set{C}|)$, where $|T|$ is the trajectory length. Finally, the cost of adding noise to the beginning/terminated transitions is also $\bigo(|\Set{C}|)$. Thus, the overall computation on the user side is $\bigo(|T||\Set{C}|)$. As all the perturbation can be implemented with bit operations, the computation is almost negligible and also affordable for any location-aware devices.

\vspace{3pt}
\noindent
\textbf{Computation on data curator.} After collecting all the perturbed information from users, data curator first estimates the unbiased frequency with OUE, and calculates the quantile of length, at the cost of $\bigo(|\Set{T}||\Set{C}|)$, where $|\Set{T}|$ is the number of trajectories/users. Besides, the computational complexity of pattern estimations (\ie length distribution, mobility patterns, and beginning/terminated states) is $\bigo(|\Set{C}|)$. Moreover, the synthesis algorithm would cost $\bigo(|\Set{T}|L)$ to synthesize the same number of trajectories as real trajectories $\Set{T}$, where $L$ denotes the mean of length distribution $\Set{L}$.
Empirical experiments in Section~\ref{sec:exp-efficiency} suggest that
the synthesis process dominates
the running time of \mymethod, and \mymethod significantly outperforms the competitor by more than two orders of magnitude.
Moreover, the data curator typically has much more computing power than the mobile devices typically used by end-users,
indicating that \mymethod is more flexible than locally point-based privacy mechanisms (\ie \baseline).

\subsection{Selecting the Grid Granularity $N$}
\label{sec:guidance-of-grid}

Finally, we propose a guideline for selecting proper grid granularity $N$ in the local setting. As mentioned before, grid granularity $N$ is an important hyperparameter for trajectory representation. We follow~\cite{icde13grid} and analyze the effect of $N$ by considering the range query on trajectories, which is the most common spatial task used by various real-world applications.
The task of a range query is to retrieve all locations of trajectories that fall within the query region.
Given a rectangular shape query region, let $r$ be the portion of the entire space covered by the query region.
There are mainly two sources of errors in \mymethod when performing a range query. The first is \textit{estimation error}. In our framework, noise is added to each transition locally, and the error of estimated transition probability is in the order of $\sqrt{ne^\epsilon/(e^\epsilon-1)^2}$. Since the query covers about $rN^2$ cells, the total error introduced by perturbation in this query is in the order of $N\sqrt{nre^\epsilon/(e^\epsilon-1)^2}$.
The second is \textit{non-uniformity error} that is proportional to the number of data points in the trajectories
that fall on the boundary 
of the query region~\cite{icde13grid}. For a query region that covers $r$ portion of the entire space (\ie grid), the length of each side is
proportional to $\sqrt{r}$ of the domain length, and thus, the number of cells overlapped with the query's boundary is in the order of $N\sqrt{r}$, and the total number of points fallen on the boundary
is in the order of $N_p/N^2\times N\sqrt{r}=\sqrt{r}N_p/{N}$, where $N_p$ represents the total number of points in all synthetic trajectories. The goal is to minimize the sum of two errors:
{\setlength{\abovedisplayskip}{3pt}
\setlength{\belowdisplayskip}{3pt}
\begin{equation}\label{equ:min-grid-size}
    \mathrm{minimize} ~N\sqrt{\frac{nre^{\epsilon'}}{(e^{\epsilon'}-1)^2}} + \frac{\sqrt{r}N_p}{N},
\end{equation}}where $\epsilon'=\epsilon_2/L$ is the budget for each transition. Since we cannot obtain trajectory length $L$ before descretization, we replace it with the geographic distance, which is on the order of $L_{\Space{R}}/f$, where $L_{\Space{R}}$ is the average number of points in trajectory,
and $f$ is the sampling ratio of the device. Besides, we also use the number of points on real trajectories $|\Set{T}|L_{\Space{R}}$ to approximate $N_p$. Finally, by minimizing Equation (\ref{equ:min-grid-size}), $N$ should be set as follows:
{\setlength{\abovedisplayskip}{3pt}
\setlength{\belowdisplayskip}{3pt}
\begin{equation}\label{equ:choose-granularity}
    N = \lambda\cdot \sqrt[4]{\frac{|\Set{T}|L_{\Space{R}}(e^{\epsilon f/L_{\Space{R}}} - 1)^2}{e^{\epsilon f/L_{\Space{R}}}}},
\end{equation}}where $\lambda$ is the hyperparameter which depends on the uniformity of the points distribution in the dataset. Note that all the parameters are easily obtained from the statistics of trajectory data (such as sampling ratio, average point count, and the data size), and we also evaluate the effectiveness of this guideline in Section~\ref{sec:exp-impact-of-grid-granularity}.



\section{Experimental evaluation}\label{sec:exp}

In this section, we first introduce the detailed experimental setup.
Next, we conduct experiments on utility as well as efficiency to illustrate the superiority of \mymethod.
Then, we conduct insight studies to evaluate the impact of each component in \mymethod. Finally, we evaluate the scalability of \mymethod and its attack-resilient ability to various real-word location-based attacks.


\subsection{Experimental Setup}

\subsubsection{\textbf{Datasets}}

We conduct our experiments on four benchmark trajectory datasets. Table~\ref{tab:dataset_description} summarizes the overall statistics.
\begin{itemize}[leftmargin=*]
    \item \textbf{Oldenburg}\footnote{\url{http://iapg.jade-hs.de/personen/brinkhoff/generator/}} is a synthetic dataset simulated by Brinkhoff's network-based moving objects generator. We generate 500,000 trajectories based on the map of Oldenburg city.
    \item \textbf{Porto}\footnote{\url{http://www.geolink.pt/ecmlpkdd2015-challenge/dataset.html}} contains taxi traces over 8 months in the city of Porto. We extract 361,591 trajectories from the central areas.
    \item {\textbf{Hangzhou}} is a private trajectory database which consists of the trace of taxis in Hangzhou city.
    \item \textbf{Campus}\footnote{\url{https://github.com/UBCGeodata/ubc-geospatial-opendata}} contains 434 buildings of British Columbia campus. We follow~\cite{vldb21ngram} to generate 1 million trajectories based on the campus buildings.
\end{itemize}

\subsubsection{\textbf{Baseline}}\label{sec:baseline}

We compare our method with \baseline~\cite{vldb21ngram}, which is the state-of-the-art (and the only) private trajectory publication method that satisfies rigorous LDP. Different from our approach, \baseline is a point-based perturbation model which directly resembles user's trajectory in the local setting. For fair comparison, we discard the auxiliary temporal and POI information, and only leverage physical distance 
to ensure the closeness between original
and sampled n-grams in geospatial space. 

\begin{table}[t]
    \centering
    \caption{Statistics of the datasets used in our experiments.}
    \resizebox{0.46\textwidth}{!}{
    \begin{tabular}{c|c|c|c}
    \hline
         \textbf{Dataset} &\textbf{Size} &\textbf{Average Length} &\textbf{Sampling Interval} \\ \hline
         \textbf{Oldenburg} &500,000 &69.75 &~15.6 sec \\ \hline
         \textbf{Porto} &361,591 &34.13 &~15 sec\\ \hline
         \textbf{Hangzhou} &348,144 &125.02 &~5 sec \\ \hline
         \textbf{Campus} &1,000,000 &35.98 &~25 sec \\ \hline
    \end{tabular}
    }
    \label{tab:dataset_description}
\end{table}

\subsubsection{\textbf{Experimental Settings}}

Based on the guideline described in
Section~\ref{sec:guidance-of-grid},
the grid granularity parameter $N$ is set to 6 for Oldenburg, Porto, and Campus dataset, and 8 for Hangzhou dataset. As for the $k$ quantile of estimated length distribution, we set it to 0.9 for all the experiments. We generate synthetic database $\Set{T}_{syn}$ with cardinality $|\Set{T}_{syn}| = |\Set{T}|$ for utility comparison. We set $\alpha=0.3$ and $\beta=0.2$ for the reweighting function defined in Equation (\ref{equ:reweight}). We set $\lambda = 2.5$ for selecting the grid granularity. As for the query region, we set $r$ as the $1/9$ proportion of the entire space.
Our experiments are conducted on a computer with Intel Xeon 2.1GHz CPU and 32 GB main memory.

\begin{table*}[t]
\centering
\caption{Utility performance comparison. The best result in each category is shown in bold. For Kendall-tau and FP F1 Similarity, higher values are better. For remaining metrics, lower values are better.}
\label{tab:overall-performance}
\resizebox{0.98\textwidth}{!}{
\begin{tabular}{c|c|c c c| c c c| c c c| c c c}
\hline
\multicolumn{2}{c|}{} &
\multicolumn{3}{c|}{\textbf{Oldenburg}} &
\multicolumn{3}{c|}{\textbf{Porto}} &
\multicolumn{3}{c|}{\textbf{Hangzhou}} &
\multicolumn{3}{c}{\textbf{Campus}}\\
\multicolumn{2}{c|}{} &
\multicolumn{1}{c}{$\epsilon=0.5$} &
\multicolumn{1}{c}{$\epsilon=1.0$} &
\multicolumn{1}{c|}{$\epsilon=1.5$} &
\multicolumn{1}{c}{$\epsilon=0.5$} &
\multicolumn{1}{c}{$\epsilon=1.0$} &
\multicolumn{1}{c|}{$\epsilon=1.5$} &
\multicolumn{1}{c}{$\epsilon=0.5$} &
\multicolumn{1}{c}{$\epsilon=1.0$} &
\multicolumn{1}{c|}{$\epsilon=1.5$} &
\multicolumn{1}{c}{$\epsilon=0.5$} &
\multicolumn{1}{c}{$\epsilon=1.0$} &
\multicolumn{1}{c}{$\epsilon=1.5$}
\\ \hline
\multirow{2}*{\textbf{Density Error}}
&\baseline      &0.0323  &0.0301 &0.0274     &0.3311 &0.3218 &0.3049     &0.1930 &0.1847 &0.1831      &0.1184 &0.1172 &0.1105\\
&\mymethod   &\textbf{0.0094}  &\textbf{0.0077}      &\textbf{0.0085}       &\textbf{0.0090}  &\textbf{0.0081} &\textbf{0.0069}  &\textbf{0.0210}     &\textbf{0.0194}            &\textbf{0.0193}  &\textbf{0.0042} &\textbf{0.0043} &\textbf{0.0037}\\ \hline
\multirow{2}*{\textbf{Query Error}}
&\baseline      &0.3362  &0.3321 &0.3206      &0.8171 &0.8102 &0.8055     &0.6454 &0.6411 &0.6332  &0.7605 &0.7579 &0.7611\\
&\mymethod   &\textbf{0.2691}  &\textbf{0.2595}      &\textbf{0.2522}        &\textbf{0.3520}  &\textbf{0.3312} &\textbf{0.3011}  &\textbf{0.2933}     &\textbf{0.2814}      &\textbf{0.2792}        &\textbf{0.2014} &\textbf{0.2011} &\textbf{0.1980}\\ \hline
 \multirow{2}*{\textbf{Hotspot Query Error}}
&\baseline      &0.2972  &0.2972 &0.2972      &1.0000 &1.0000 &1.0000     &0.1529 &0.1529 &0.1529     &0.7001 &0.7001 &0.7001\\
&\mymethod   &\textbf{0.0593}   &\textbf{0.0593}      &\textbf{0.0530}        &\textbf{0.0000}  &\textbf{0.0000} &\textbf{0.0000} &\textbf{0.0131}     &\textbf{0.0131}      &\textbf{0.0131}      &\textbf{0.0013} &\textbf{0.0013} &\textbf{0.0000}\\
 \hline
\multirow{2}*{\textbf{Kendall-tau}}
&\baseline      &0.7479  &0.7512 &0.7498      &0.5442 &0.5791 &0.5934     &0.6164 &0.6242 &0.6623       &0.3794 &0.3841 &0.3852\\
&\mymethod   &\textbf{0.8874}  &\textbf{0.8944}      &\textbf{0.8942}       &\textbf{0.6672}  &\textbf{0.7114} &\textbf{0.7581}  &\textbf{0.6782}     &\textbf{0.7044}      &\textbf{0.7174}      &\textbf{0.8603}  &\textbf{0.8574}  &\textbf{0.8667}\\ \hline
\multirow{2}*{\textbf{Trip Error}}
&\baseline      &0.1251  &0.1231 &0.1230    &0.5122 &0.5043 &0.4979     &0.4493 &0.4490 &0.4421       &0.3052 &0.3024 &0.2995\\
&\mymethod   &\textbf{0.0697}  &\textbf{0.0683}      &\textbf{0.0672}        &\textbf{0.0762}  &\textbf{0.0778} &\textbf{0.0771} &\textbf{0.0531}     &\textbf{0.0513}      &\textbf{0.0504}      &\textbf{0.0744} &\textbf{0.0740} &\textbf{0.0729}\\
 \hline
\multirow{2}*{\textbf{Length Error}}
&\baseline      &0.1142  &0.1134 &0.1112      &0.1812 &0.1823 &0.1791    &0.02922 &0.02736 &0.02607      &0.1102 &0.1041 &0.1045\\
&\mymethod   &\textbf{0.0373}  &\textbf{0.0370}      &\textbf{0.0379}       &\textbf{0.0410}  &\textbf{0.0399} &\textbf{0.0392}  &\textbf{0.0036}     &\textbf{0.0037}      &\textbf{0.0032}       &\textbf{0.0984}  &\textbf{0.0981}  &\textbf{0.0987}\\
 \hline
\multirow{2}*{\textbf{Diameter Error}}
&\baseline      &0.1244  &0.1242 &0.1221      &0.2201 &0.2180 &0.2174     &0.2013 &0.1989 &0.1972      &0.0682 &0.0663 &0.0658 \\
&\mymethod   &\textbf{0.0568}  &\textbf{0.0570}      &\textbf{0.0562}        &\textbf{0.0354}  &\textbf{0.0340} &\textbf{0.0331}  &\textbf{0.0572}     &\textbf{0.0569}      &\textbf{0.0563}      &\textbf{0.0591} &\textbf{0.0589} &\textbf{0.0587}\\
 \hline
\multirow{2}*{\textbf{Pattern F1}}
&\baseline      &0.33  &0.33 &0.34     &0.19 &0.18 &0.19     &0.24 &0.26 &0.26      &0.32 &0.32 &0.33 \\
&\mymethod   &\textbf{0.69}  &\textbf{0.69}      &\textbf{0.69}        &\textbf{0.67}  &\textbf{0.63} &\textbf{0.65}  &\textbf{0.80}     &\textbf{0.80}     &\textbf{0.80}      &\textbf{0.71} &\textbf{0.71} &\textbf{0.72}\\
 \hline
\multirow{2}*{\textbf{Pattern Error}}
&\baseline      &0.8033  &0.8004 &0.7982      &0.9206 &0.9232 &0.9181     &0.8782 &0.8725 &0.8754      &0.7992 &0.7912 &0.7789\\
&\mymethod   &\textbf{0.5693}  &\textbf{0.5632}      &\textbf{0.5594}        &\textbf{0.6498}  &\textbf{0.6687} &\textbf{0.6692} &\textbf{0.4593}     &\textbf{0.4552}      &\textbf{0.4554}      &\textbf{0.5502} &\textbf{0.5508} &\textbf{0.5501}\\
 \hline

\end{tabular}}
\end{table*}

\subsection{Utility Metrics}
To comprehensively quantify the utility of the synthetic trajectories, we adopt various utility metrics from three categories, including global level, trajectory level, and semantic level.

\vspace{3pt}
\noindent
\textbf{Global level utility} measures the spatial patterns of trajectories in a global view, which serves as a building block of various spatial applications like range query and traffic forecasting. We use the following metrics for evaluation:
\begin{itemize}[leftmargin=*]
    \item \textbf{Density error} evaluates the density difference between synthetic trajectory set $\Set{T}_{syn}$ and real trajectory set $\Set{T}$. 
    \begin{equation}
        Density\ Error=JSD\big(\mathcal{D}(\Set{T}), \mathcal{D}(\Set{T}_{syn})\big),
    \end{equation}
    where $\mathcal{D}(\Set{P})$ denotes the grid density distribution in a given set $\Set{P}$, and $JSD(\cdot)$ represents the Jenson-Shannon divergence between two distributions.
    \item\textbf{Query error} is a popular measure for evaluating data synthesis algorithms ranging from tabular data to graph and location data~\cite{ccs12ngram,vldbj14network,edbt14}. We consider range queries
    of trajectories in a random spatial region $R$, \ie $Q(\Set{P})$ returns the number of points in any trajectory of a specified set $\Set{P}$ that are within the spatial region $R$.
    {\setlength{\abovedisplayskip}{3pt}
    \setlength{\belowdisplayskip}{3pt}
    \begin{equation}
        Query\ Error=\frac{|Q(\Set{T})-Q(\Set{T}_{syn})|}{max\{Q(\Set{T}), z)\}},
    \end{equation}}where $z$ is the sanity bound to weaken the influence of queries that return very small counts. 
    We set the sanity bound $z$ to $\frac{\sum_{\Set{T}}|T|}{100}$, 
    and report the average result of 200 random queries.
    \item\textbf{Hotspot query error} measures the ability of preserving spatial hotspots. Specifically, we choose top-$n_h$ mostly visited cells in $\mathcal{T}$ and $\mathcal{T}_{syn}$ as the hotspots set $H$ and $H_{syn}$, respectively.
    {\setlength{\abovedisplayskip}{3pt}
    \setlength{\belowdisplayskip}{3pt}
    \begin{equation}
        HQE=1-\frac{\sum_{C_i\in H_{syn}} rel(C_i)/\log(rank_{H_{syn}}(C_i)+1)}{\sum_{j=1}^{n_h} 1/(j\cdot \log(j+1))}, 
    \end{equation}}where $rank_{H_{syn}}(C_i)$ is the position of $C_i$ in $H_{syn}$, $rel(C_i)$ is the relativity score of $C_i$: when $C_i\in H$, $rel(C_i) = 1/rank_{H}(C_i)$, else  $rel(C_i) = 0$. We set $n_h=5$.
    \item\textbf{Kendall's tau coefficient} is for modelling the discrepancies in locations’ popularity ranking~\cite{ccs18adatrace}. Let $\mathcal{D}(C_i)$ be the density of cell $C_i$, and $(C_i,C_j)$ be a \textit{concordant pair} if and only if $\mathcal{D}(C_i)\ge\mathcal{D}(C_j)$ or $\mathcal{D}(C_i)\le\mathcal{D}(C_j)$ holds both on $\Set{T}$ and $\Set{T}_{syn}$. Otherwise, it's a \textit{discordant pair}.
    {\setlength{\abovedisplayskip}{3pt}
    \setlength{\belowdisplayskip}{3pt}
    \begin{equation}
        Kendall\text{-}tau=\frac{N_c-N_d}{|\Set{C}|(|\Set{C}|-1)/2},
    \end{equation}}where $N_c$ and $N_d$ represent the number of concordant pairs and the number of discordant pairs respectively, and $|\Set{C}|$ captures the total number of grid cells.
\end{itemize}

\vspace{3pt}
\noindent
\textbf{Trajectory level utility} denotes the spatial features within each trajectory, which also fascinates a wild range of applications like origin-destination analysis and commutes studies. We employ the following three evaluation metrics:
\begin{itemize}[leftmargin=*]
    \item  \textbf{Trip error} measures how well the correlations between trips’ starting points and ending points are preserved~\cite{ccs18adatrace}. Specifically, we calculate the probability distribution of start/end points in $\Set{T}$ and that in $\Set{T}_{syn}$, and utilize Jensen-Shannon divergence to measure their difference.
    \item \textbf{Length error} focuses on the difference between real and synthetic datasets in terms of trajectory lengths (\ie distance travelled by each trajectory), which calculates the total distance in a trajectory by adding up the Euclidean distance between consecutive points.
    \item \textbf{Diameter error} is defined as the difference of maximum distance frequency between real and synthetic datasets, where the maximum distance refers to the maximum Euclidean distance between two points in a trajectory~\cite{vldb15dpt,icde22gl}.
\end{itemize}
Since travel distance and diameter are continuous, we follow~\cite{ccs18adatrace} to separate them into 20 equi-width buckets, and calculate the distribution of the obtained histogram. We then use Jenson-Shannon divergence to quantify the above errors.

\begin{table*}[t]
    \centering
    \caption{Average runtime in seconds. We report the average running time per 1,000 trajectories of each component.}
    \label{tab:time-efficiency}
    \resizebox{0.9\textwidth}{!}{

    \begin{tabular}{c|ccc|ccc|c||c}
    \hline
       & \multicolumn{3}{c|}{\textbf{\mymethod : User side}}   & \multicolumn{3}{c|}{\textbf{\mymethod : Curator side}}    &\multirow{2}{*}{\textbf{Total}} &\multirow{2}{*}{\textbf{\baseline Total}} \\
      & \textbf{Length} &\textbf{Intra-traj tran.} &\textbf{Beg./Ter. tran.} &\textbf{Preparation} &\textbf{Pattern est.} &\textbf{Synthesis} & & \\ \hline
     \textbf{Oldenburg} &0.030  &0.192  &0.104 &0.010 &0.001 &0.219 &\textbf{0.556} &183 \\
     \textbf{Porto} &0.028 &0.148 &0.092 &0.015 &0.001 &0.179 &\textbf{0.463} &132  \\
     \textbf{Hangzhou} &0.028 &0.302 &0.098 &0.012 &0.002 &0.304 &\textbf{0.746} &418 \\ 
     \textbf{Campus} &0.028 &0.138 &0.083 &0.009 &0.001 &0.189 &\textbf{0.448} &197 \\\hline

    \end{tabular}}
\end{table*}

\vspace{3pt}
\noindent
\textbf{Semantic level metrics.} Apart from the aforementioned statistic metrics of trajectories' attributes, we also adopt two semantic level metrics to mine the mobility patterns that are hidden behind those statistics. A pattern $P$ is defined as an ordered sequence of consecutive cells, and we select top-$n$ most occurred patterns in $\Set{T}$ and $\Set{T}_{syn}$, denoted as pattern sets $FP$ and $FP_{syn}$ respectively, and calculate the following metrics:
\begin{itemize}[leftmargin=*]
    \item \textbf{Pattern F1} evaluates the similarity between the selected most frequent pattern sets $FP$ and $FP_{syn}$:
    {\setlength{\abovedisplayskip}{3pt}
    \setlength{\belowdisplayskip}{3pt}
    \begin{equation}
        Pattern\ F1=2\times\frac{Precision(FP, FP_{syn})\times Recall(FP, FP_{syn})}{Precision(FP, FP_{syn})+ Recall(FP, FP_{syn})}.
    \end{equation}}
    \item\textbf{Pattern Error} measures the relative difference between the number of pattern occurrences in each dataset:
    {\setlength{\abovedisplayskip}{3pt}
    \setlength{\belowdisplayskip}{3pt}
    \begin{equation}
        Pattern\ Error = \frac{1}{|FP|}\textstyle\sum_{P\in FP}\frac{|n^P-n^P_{syn}|}{n^P},
    \end{equation}}where $n^P$/$n^P_{syn}$ is the number of occurrence of pattern $P$ in the dataset $FP$/$FP_{syn}$, and we use top 100 frequent patterns (\ie $|FP|=100$) for evaluation.
\end{itemize}

\begin{table*}[t]
    \centering
    \caption{Impact of beginning/terminated transitions. Best result is shown in bold. 
    HQ Error denotes ``Hotpot Query Error''.
    }
    \label{tab:biginning/erminated-transitions}
    \resizebox{0.97\textwidth}{!}{
    \begin{tabular}{c|c|c|c|c|c|c|c|c|c|c}
    \hline
      \textbf{Dataset} &\textbf{Model}  &\textbf{Density Error} &\textbf{Query Error} &\textbf{HQ Error} &\textbf{Kendall-tau}  &\textbf{Trip Error}  &\textbf{Length Error}  &\textbf{Diameter Error}  &\textbf{Pattern F1}  &\textbf{Pattern Error}  \\ \hline
      \multirow{4}*{\textbf{Oldenburg}}
      &\randsyn   &0.0569 &0.7491 &0.0724 &0.6623 &0.1779 &0.0694 &0.0611 &0.58 &\textbf{0.4313} \\
      &\combtran &0.0168 &0.2794 &0.0593 &0.8108 &0.1052 &0.0371 &\textbf{0.0570} &0.67 &0.5596 \\
      &\noadapt &0.0081 &0.3754 &\textbf{0.0013} &0.8866 &0.0971 &0.0713 &0.0932 &\textbf{0.69} &0.6822 \\
      &\mymethod &\textbf{0.0077} &\textbf{0.2595} &0.0593 &\textbf{0.8944} &\textbf{0.0683} &\textbf{0.0370} &\textbf{0.0570} &\textbf{0.69} &0.5632 \\ \hline
      \multirow{4}*{\textbf{Porto}}
      &\randsyn   &0.2687 &7.4933 &0.2583 &0.2191 &0.4372 &0.1244 &0.1460 &0.39 &0.6658 \\
      &\combtran &0.0243 &0.6741 &0.0464 &0.5784 &0.0932 &0.0412 &0.0351 &0.59 &0.6730 \\
      &\noadapt &0.0098 &0.3607 &0.0464 &0.6828 &0.1052 &0.0579 &0.0631 &\textbf{0.63} &0.7674 \\
      &\mymethod &\textbf{0.0081} &\textbf{0.3312} &\textbf{0.0000} &\textbf{0.7114} &\textbf{0.0778} &\textbf{0.0399} &\textbf{0.0340} &\textbf{0.63} &0.6687 \\ \hline
      \multirow{4}*{\textbf{Hangzhou}}
      &\randsyn   &0.1928 &2.4842 &0.0144 &0.5252 &0.4471 &0.0549 &0.1344 &0.49 &\textbf{0.3821}\\
      &\combtran &0.0361 &0.3426 &0.0464  &0.6492 &0.0762 &0.0040 &0.0594 &0.78 &0.4890 \\
      &\noadapt &0.0322 &0.3641 &\textbf{0.0131} &\textbf{0.7062} &0.0593 &0.0041 &0.0921 &\textbf{0.80} &0.5866 \\
      &\mymethod &\textbf{0.0194} &\textbf{0.2814} &\textbf{0.0131}  &0.7044 &\textbf{0.0513} &\textbf{0.0037} &\textbf{0.0569} &\textbf{0.80} &0.4552 \\ \hline
      \multirow{4}*{\textbf{Campus}}
      &\randsyn   &0.0914 &1.5203 &0.3055 &0.5587 &0.2811 &0.0982 &0.0589  &0.64 &\textbf{0.4109}\\
      &\combtran &0.0103 &0.3025 &0.0304 &0.8539 &0.0862 &\textbf{0.0973} &0.0593  &0.71 &0.5780 \\
      &\noadapt &0.0046 &0.3766 &0.0593 &\textbf{0.8571} &\textbf{0.0502}  &0.1496 &0.0590&\textbf{0.71} &0.6892 \\
      &\mymethod &\textbf{0.0043} &\textbf{0.2011}&\textbf{0.0013}  &\textbf{0.8571} &0.0740 
      &0.0981 &\textbf{0.0587}&\textbf{0.71}  &0.5508 \\ \hline

    \end{tabular}}

\end{table*}

\subsection{Utility Evaluation}\label{sec:exp-utility}

In our first set of experiments, we compare the utility performance of \mymethod and \baseline with various privacy budgets $\epsilon$. For each experiment, we perform the synthesis 5 times, and report the average results in Table~\ref{tab:overall-performance}. Generally speaking, \mymethod outperforms \baseline in all utility metrics across all datasets, which well demonstrates the robustness and strong utility-preserving ability of \mymethod. Based on the in-depth analysis of the results, we have made the following detailed observations.
\begin{itemize}[leftmargin=*]
    \item We first analyze 
    the performance \wrt different types of utility metrics. Since \baseline fails to take the global features of trajectories into consideration, it suffers from severe performance degradation on the global level utility. Especially, the density error in \mymethod is \emph{three times} smaller than that in \baseline, which implies that the trajectories generated by \baseline fail to preserve the geospatial density due to the point-based perturbation. Also, our method shows strong performance in 
    both the trajectory level and the semantic level utilities. We contribute the improvement to the mobility modelling and adaptive synthesis of \mymethod: (1) By collecting transitions from both intra-trajectory and start/end points, \mymethod is able to synthesize trajectories whose mobility patterns (\eg moving directions) 
    assemble users' real traces. In contrast, \baseline aims to preserve the local $n$-gram closeness when perturbing, but ignores the critical sequential features of trajectories in a global view. (2) Benefited from the adaptive synthesis algorithm, \mymethod can generate authentic trajectories with proper length to maintain the relationship between start points and end points, while the point-based \baseline fails to capture the long-term dependence between beginning and terminated transitions.
    \item 
    Next, we evaluate the effectiveness \wrt different datasets. We find that the performance of our framework is robust, which keeps good utility on both synthetic and real-world datasets. However, the performance of \baseline varies across datasets. Specifically, it achieves competitive results in terms of density error and trip error on Oldenburg, but the results of these metrics are much worse in the two real datasets.
    We argue that this is because the mobility patterns of real-world trajectories are much more complex than synthetic ones, and \baseline is unable to capture these patterns with n-gram model.
    \item We also examine the effects \wrt different privacy budgets $\epsilon$. It is worth mentioning that \mymethod could achieve good performance even when the budget is small (\eg $\epsilon=0.5$), which confirms the effectiveness of capturing trajectories' key patterns with only a few distributions. However, \baseline relies on a large budget to achieve reasonable performance, meaning that it needs to relax the privacy guarantee for practical use. For privacy budget $\epsilon$, the protection is acceptable as long as the budget is less than 2~\cite{ccs14rappor,apple, usenix17oue,vldb22query}.
    To align with other real-world deployments of LDP by Google~\cite{ccs14rappor} and Microsoft~\cite{nips17microsoft}, we set $\epsilon=1$ in the following experiments.

\end{itemize}

\subsection{Efficiency Evaluation}\label{sec:exp-efficiency}

Since efficiency is equally important as utility for real-world deployments, we conduct comprehensive experiments to evaluate the average running time of each component, which is detailed in Table~\ref{tab:time-efficiency}. Generally speaking, \mymethod is a very efficient privacy-preserving trajectory publication framework, which is more than \emph{300 times} faster than \baseline. The reason is that \baseline suffers from time-consuming processes like solving linear programming problem and recursive reconstructions, not to mention the expensive cost of pre-processing for POIs and other external knowledge, which greatly hinders its use on mobile devices. We would like to highlight that this observation is consistent with the authors’ claims in their paper~\cite{vldb21ngram}, as they also report up to 4 seconds per generated trajectory. 

As for the breakdown of time spent on each component,  we find that the major computation at user side is the perturbation of intra-trajectory transitions because it has to perturb and upload multiple times to report the holistic transition states. 
Other one-time operations like trajectory length perturbation are much faster, 
consistent with the computational analysis presented in Section~\ref{sec:computationl-cost}.
On the data curator side, the main computation is incurred by the synthesis process, since it needs to generate location points sequentially according to the transition probability.
Overall, \mymethod is highly efficient and practical for real-world applications, and the time cost for privacy protection is nearly imperceptible for users.

\begin{figure*}[t]
    \centering
    \includegraphics[width=0.26\textwidth]{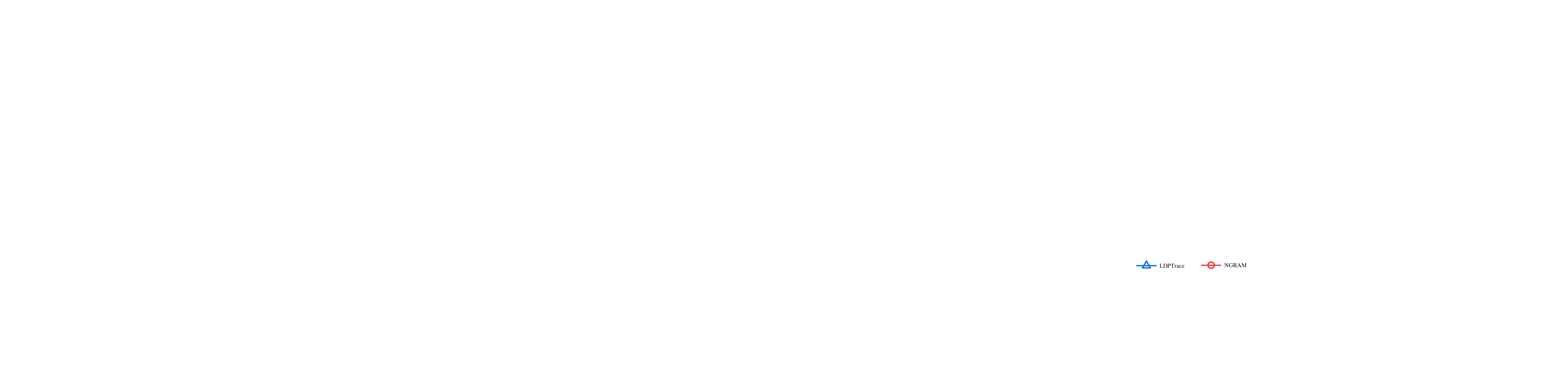}
    \includegraphics[width=0.92\textwidth]{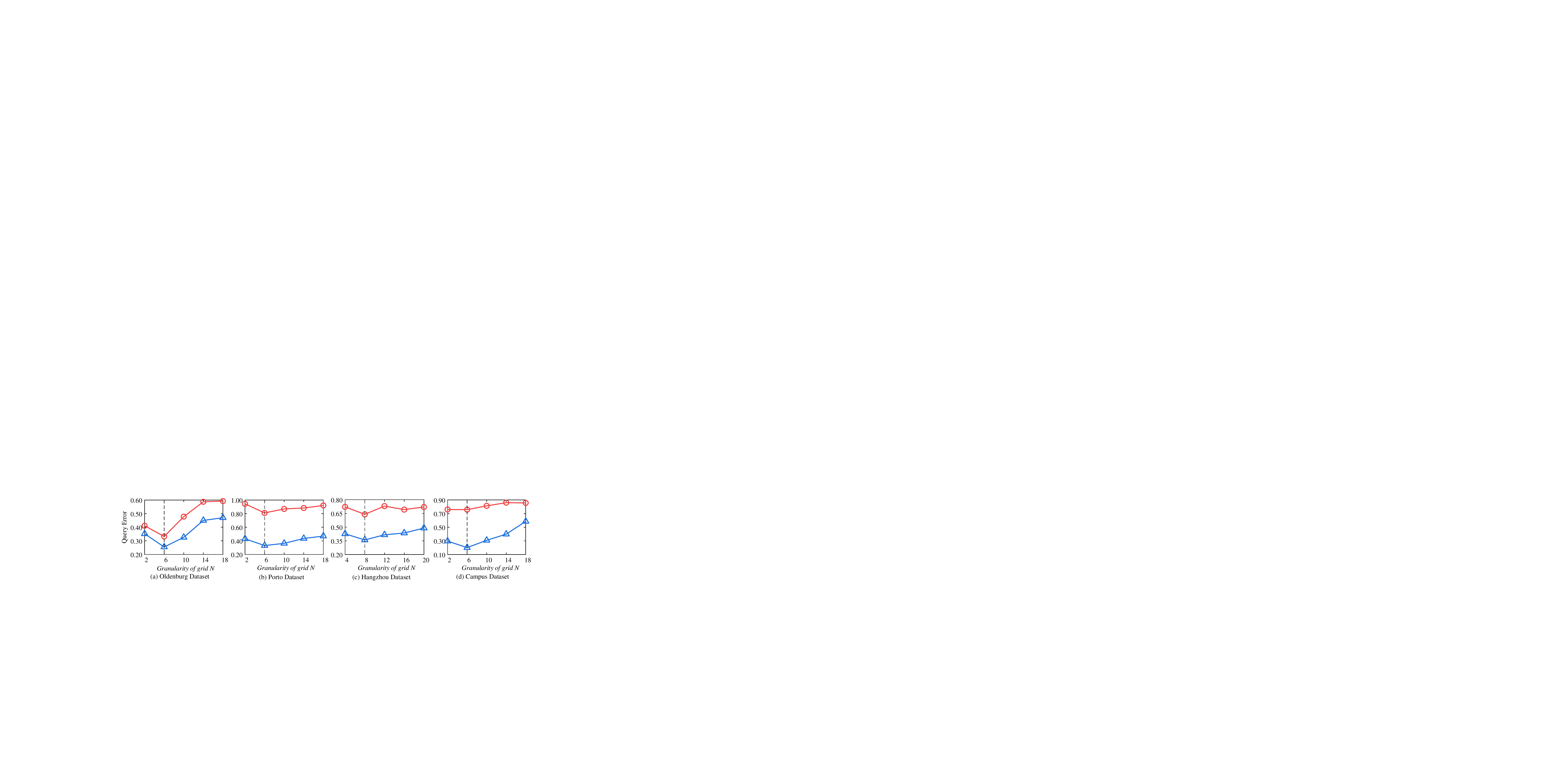}
    \caption{Impact of grid granularity $N$ with optimal $N$ values derived based on Equation~\eqref{equ:choose-granularity} represented by dotted lines.
    }
    \label{fig:granularity}
\end{figure*}

\begin{figure}[t]
    \centering
    \includegraphics[width=0.499\textwidth]{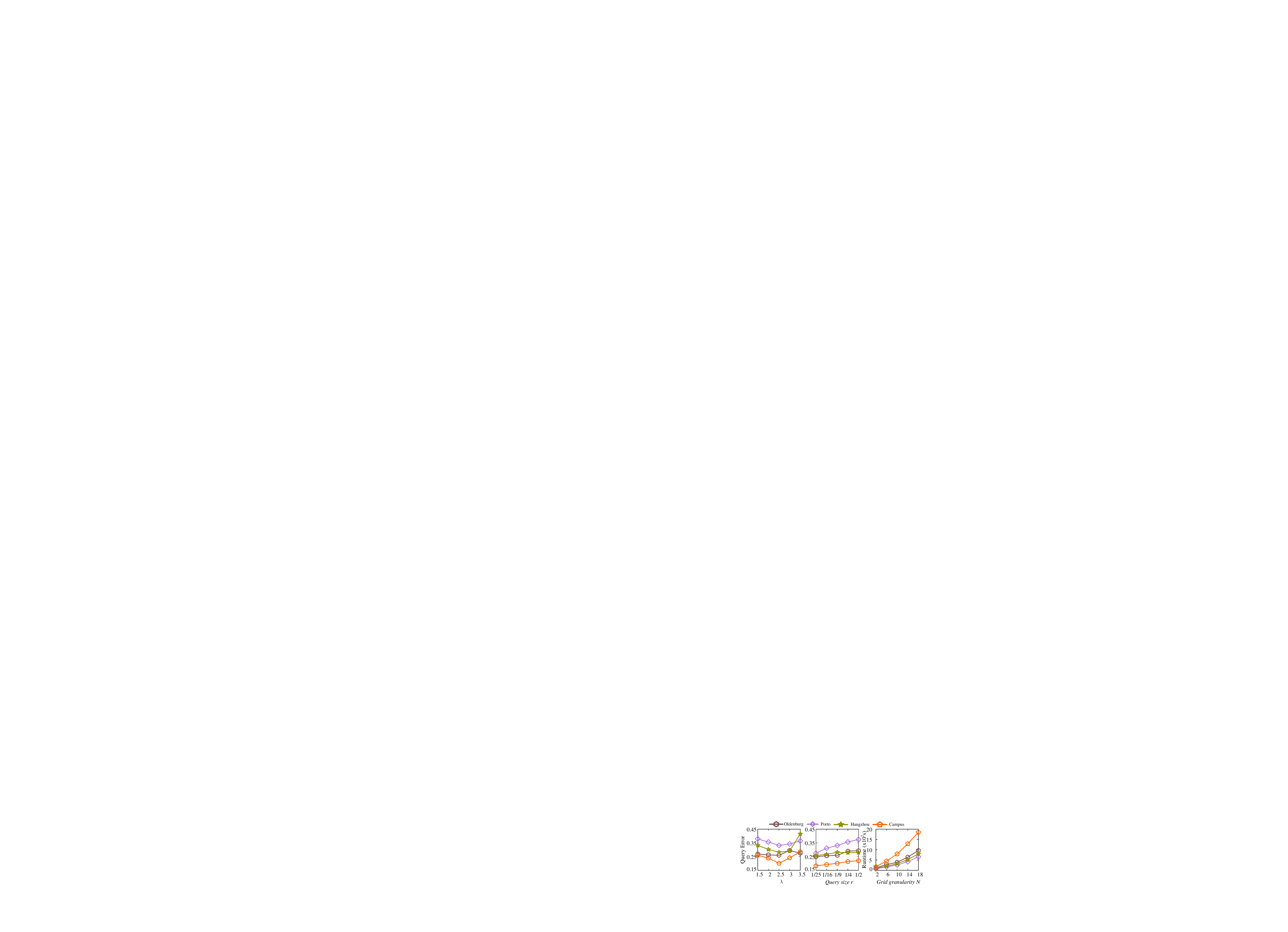}
    \caption{(a) Impact of $\lambda$. (b) Impact of query size $r$. (c) Impact of granularity $N$ on total runtime (seconds).}
    \label{fig:lambda_size_runtime}
\end{figure}


\subsection{Analysis of \mymethod} \label{sec:ablation-study}

As the pattern modelling is at the core of \mymethod, we also conduct insight experiments to investigate its effectiveness, \ie how the presence of beginning/terminated transitions, grid granularity, and the adaptive synthesis algorithm affect our model.

\subsubsection{\textbf{Impact of beginning/terminated transitions.}} We first verify the effectiveness of the beginning/terminated transitions. To this end, we construct three variants of \mymethod, including 
(1) \randsyn that discards the virtual start/end points from \mymethod, and 
synthesizes trajectories according to intra-transition probabilities,
(2) \combtran that combines the beginning/terminated transitions with intra-transitions, and perturbs/reports them as a whole,
and (3) \noadapt that removes the adaptive synthesis strategy (\ie Equation~\eqref{equ:reweight}).
We summarize the results in Table~\ref{tab:biginning/erminated-transitions}.

Compared with the complete model \mymethod, the absence of the beginning/terminated transitions (\ie \randsyn) dramatically degrades the utility, indicating the necessity of modelling the start/end point distribution. However, it is noticed that the pattern error of \randsyn remains small, since it measures the frequent intra-mobility patterns that are less irrelevant with endpoints. \combtran introduces unnecessary noise for estimating beginning/terminated transitions as they will be perturbed along the intra-transitions. Thus, it also results in some utility loss. Last but not the least, directly synthesizing trajectory without considering current length will make the synthetic trajectory too short to represent useful spatial patterns, incurring suboptimal performance.

\subsubsection{\textbf{Impact of grid granularity $N$}} \label{sec:exp-impact-of-grid-granularity}

We analyze the influence of different grid granularity settings to empirically verify the effectiveness of our guideline for choosing $N$. Specifically, we utilize the query error metric to measure the impact of grid granularity. Figure~\ref{fig:granularity} depicts the results, consistent with the analysis presented in Section~\ref{sec:guidance-of-grid}.
We can observe that the performance reaches near optimal at the estimated granularity which is derived from Equation~(\ref{equ:choose-granularity}), \ie 8 for Hangzhou dataset and 6 for the other three datasets.
In addition, we see similar trend with respect to the performance of \baseline. The reason behind is that our estimation only depends on the statistics of trajectory dataset, which is model-agnostic, and can be a good reference to choose the grid granularity for all locally private methods.

\begin{figure}[t]
    \centering
    \includegraphics[width=0.22\textwidth]{figures/legend_model.pdf}
    \includegraphics[width=0.43\textwidth]{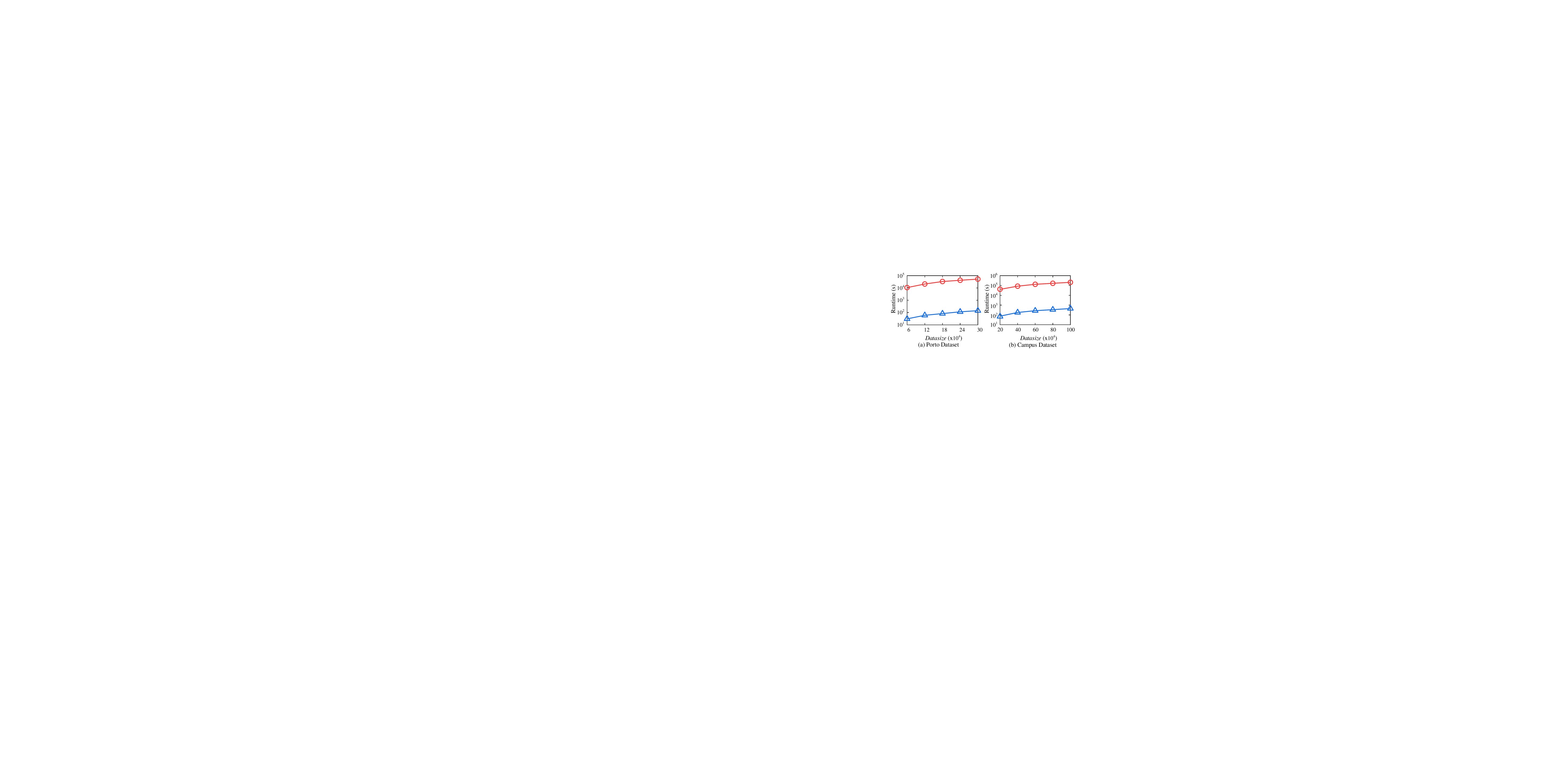}
    \caption{Scalability evaluation.}
    \label{fig:scalability}
\end{figure}

Besides, since the granularity will influence the trajectory length and the number of transition states, we also test the running time of \mymethod under different $N$ values in  Figure~\ref{fig:lambda_size_runtime}(c). There is a clear increasing trend of runtime as $N$ increases its value and each grid becomes more fine-grained.
Consequently, a wrong choice of granularity will not only cause the utility degradation, but also result in larger computational complexity.

\begin{table*}[t]
    \centering
    \caption{Utility performance comparison with unmodified \baseline on Campus dataset.}
    \label{tab:unmodified-ngram}
    \resizebox{0.99\textwidth}{!}{
    \begin{tabular}{c|c|c|c|c|c|c|c|c|c}
    \hline
\textbf{Model}  &\textbf{Density Error} &\textbf{Query Error} &\textbf{HQ Error} &\textbf{Kendall-tau}  &\textbf{Trip Error}  &\textbf{Length Error}  &\textbf{Diameter Error}  &\textbf{Pattern F1}  &\textbf{Pattern Error}  \\ \hline
      \baseline &0.1172 &0.7579 &0.7001 &0.3841 &0.3024 &0.1041 &0.0663 &0.32 &0.7912 \\
      unmodified \baseline &0.0536 &0.4371 &0.8108 &0.6349 &0.2006 &0.1122 &0.0638 &0.55 &0.8215 \\
      \mymethod &\textbf{0.0043} &\textbf{0.2011}&\textbf{0.0013}  &\textbf{0.8571} &\textbf{0.0740} &\textbf{0.0981} &\textbf{0.0587}  &\textbf{0.71}  &\textbf{0.5508} \\ \hline

    \end{tabular}}
\end{table*}

\begin{figure*}[t]
    \centering
    \includegraphics[width=0.55\textwidth]{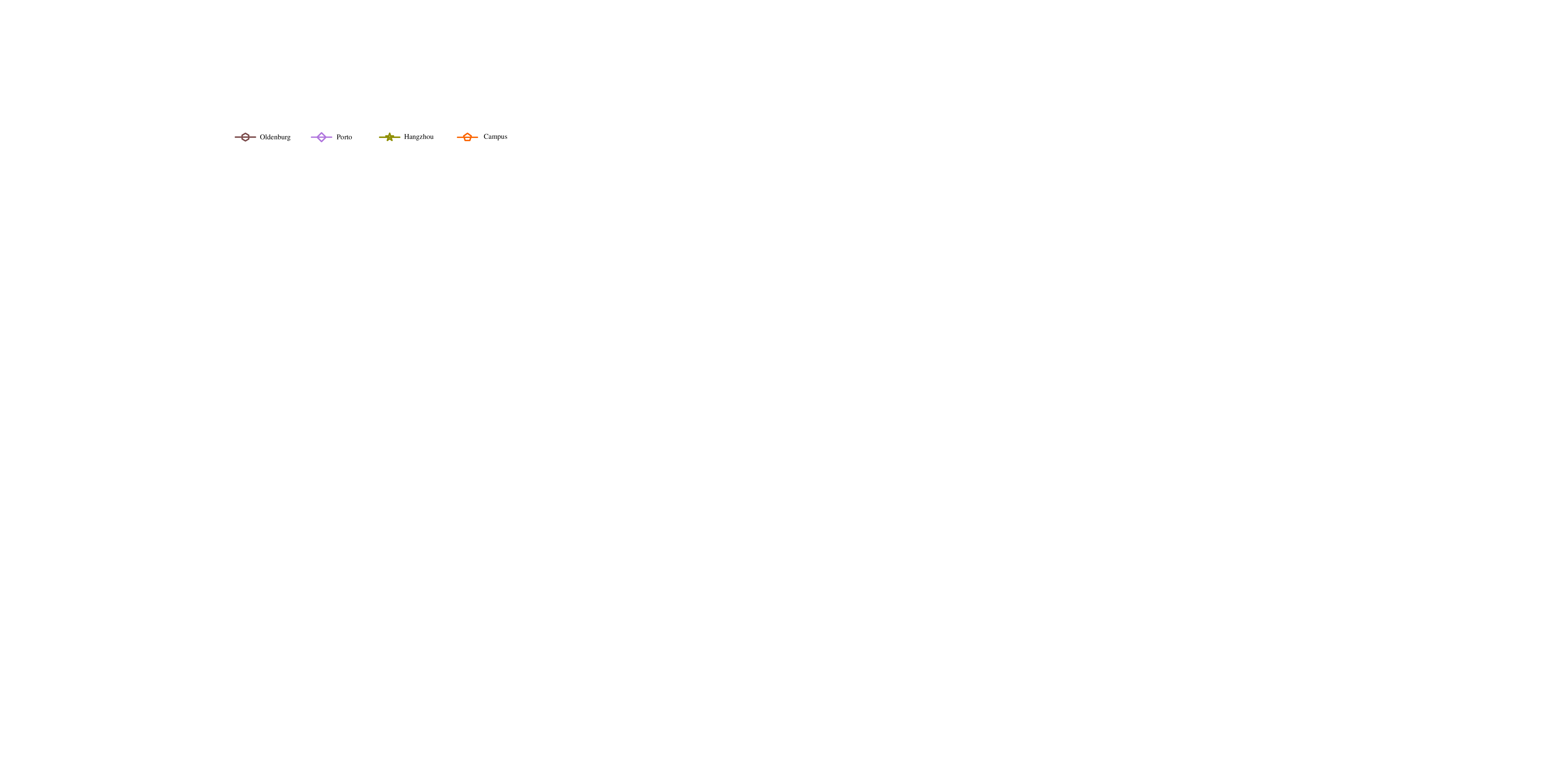}
    \includegraphics[width=0.94\textwidth]{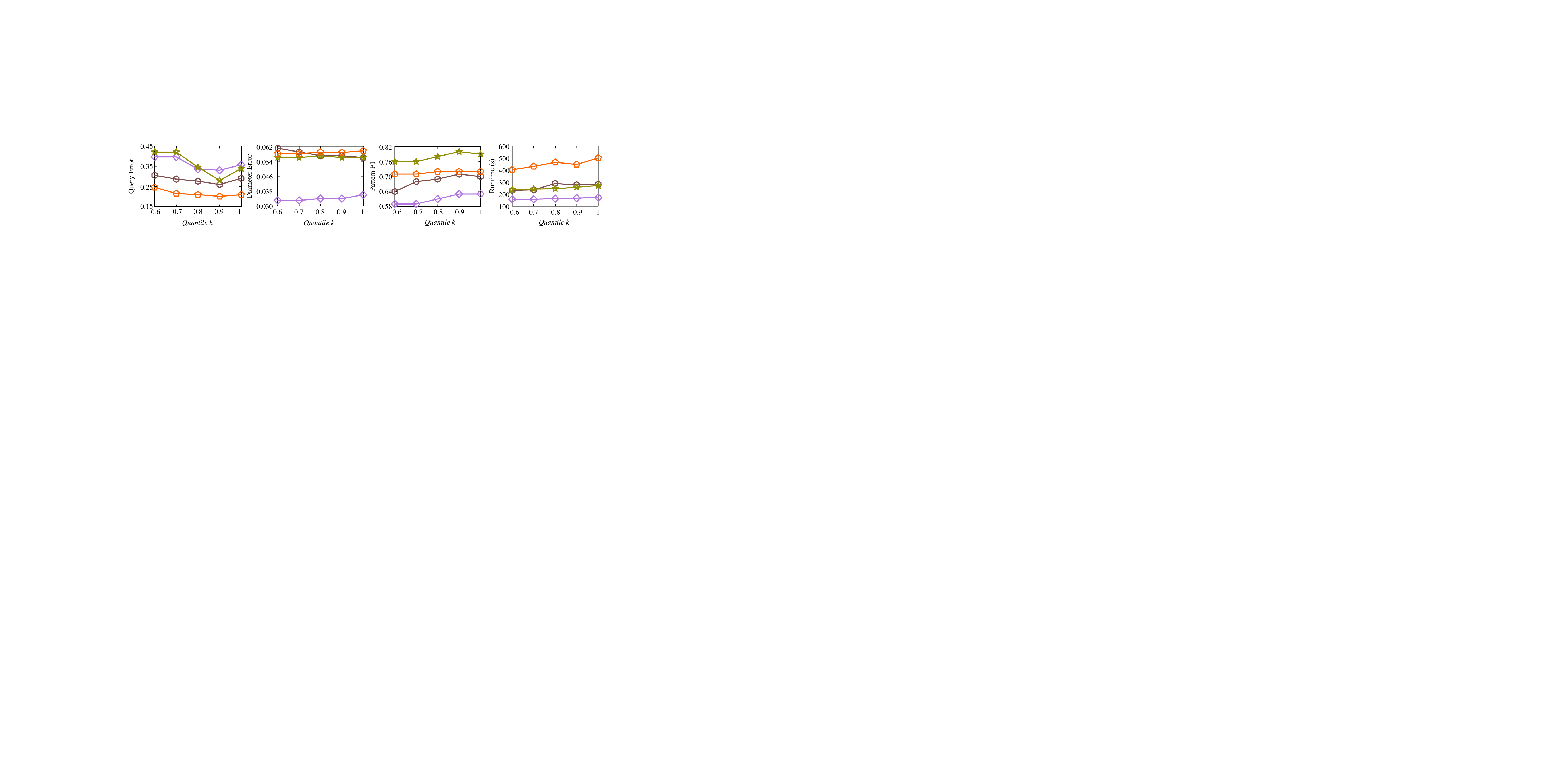}
    \caption{Impact of quantile $k$ of length distribution.}
    \label{fig:variation-k}
\end{figure*}

\subsubsection{\textbf{Impact of quantile $k$}} \label{sec:exp-impact-of-k}

For space reason, we only choose one representative utility metric each from global, trajectory, and semantic levels (\ie query error, diameter error, and pattern F1) to explore the influence of quantile $k$. In general, the selection of trajectory length has different impacts on different utilities. When the quantile $k$ is small, more transitions will be truncated, and a large bias will be introduced to intra-transition modelling, leading to a large query error. On the other hand, when $k$ becomes larger, the budget of each reported transition is smaller, and thus, the added noise increases and the synthetic trajectories may be unreliable, which results in a larger diameter error. In conclusion, we choose $k=0.9$, since it achieves a good trade-off among different utilities.


The choice of $k$ also 
impacts the efficiency of \mymethod. As illustrated in the last subfigure of Figure~\ref{fig:variation-k}, the running time grows plainly with the growth of $k$. It is consistent with our expectation because a larger $k$ requires a longer amount of time for perturbation and reporting.
However, the magnitude of the time increase is insignificant, because our algorithm is super efficient.
Hence, the influence of more perturbations is almost negligible.

\subsubsection{Impact of $\lambda$ and query size $r$.}
We conduct experiments on different $\lambda$ and query size $r$ to further analyze the effectiveness of the granularity selection method. Recall that $\lambda$ is the hyperparameter which depends on the uniformity of the points distribution in the dataset. As shown in Figure~\ref{fig:lambda_size_runtime}(a), we find that $\lambda=2.5$ achieves good performance across different datasets, since it reaches good balance between non-uniformity error and transition estimation error.

As shown in Figure~\ref{fig:lambda_size_runtime}(b), the query error increases
plainly when the query size $r$ becomes larger. There are two competing effects when increasing $r$: on the one hand, the error from more grid cells is aggregated; on the other hand, each query is less affected by noise, since actual counts are larger. The first effect is stronger, so the overall error increases with query size. The results are consistent with prior studies~\cite{vldb22query}. 

\subsection{Scalability}\label{sec:exp-scalability}

We also study the scalability of \mymethod by varying the cardinality of the trajectory datasets. We observe similar trends in all four datasets. Due to space limitation, We only report the total running time under Porto and Campus datasets in Figure~\ref{fig:scalability}.
As observed, \mymethod is consistently faster than \baseline in all different dataset scales, and has \emph{two orders of magnitude} improvement. \mymethod also has stable performance, and the processing time will not grow sharply with the growth of the dataset size. 
For \baseline, it takes more than two days in processing all the trajectories when the dataset is at the scale of millions, while \mymethod can finish the whole process in less than 10 minutes.
Therefore, \mymethod is suitable for large-scale deployment with little computational cost.

\subsection{Comparison with unmodified \baseline}
To further demonstrate the superiority of \mymethod, we also compare it against the unmodified \baseline with external knowledge. Specifically, we use the same Campus dataset described in~\cite{vldb21ngram}, and carefully attach the auxiliary knowledge (POIs, temporal information, hierarchical category) to the original trajectories. The results of unmodified \baseline are reported in Table~\ref{tab:unmodified-ngram}. As observed, \mymethod outperforms both \baseline and unmodified \baseline for all utility metrics by a large margin. Although unmodified \baseline benefits from the deterministic constraints of external knowledge to ensure the semantic similarity between original and perturbed locations, it still fails to capture the moving patterns properly. 





\begin{figure}[t]
    \centering
    \includegraphics[width=0.23\textwidth]{figures/legend_model.pdf}
    \includegraphics[width=0.43\textwidth]{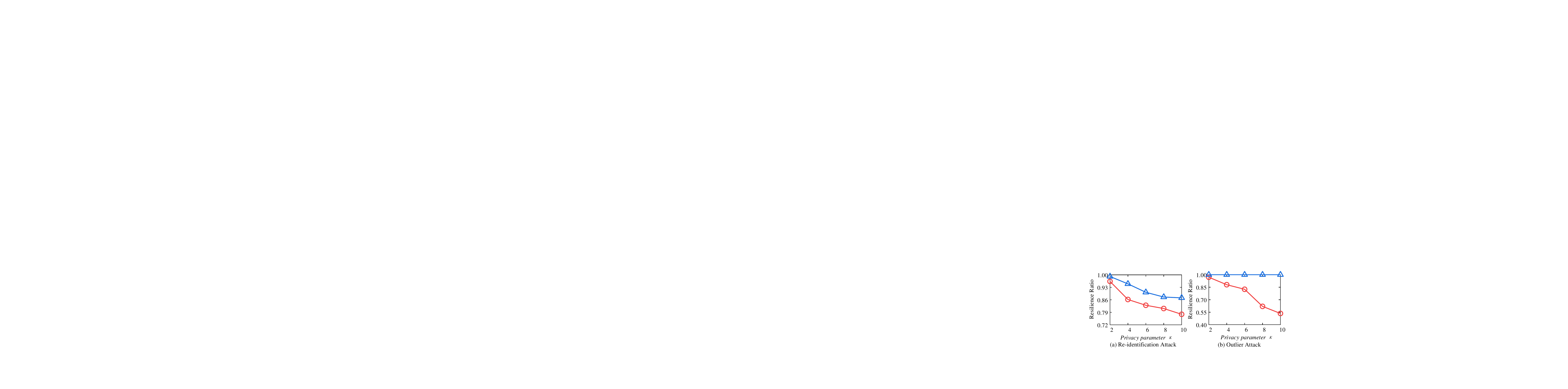}
    \caption{Attack resilience analysis on Oldenburg dataset.}
    \label{fig:attack}
\end{figure}

\subsection{Attack Resilience}\label{sec:exp-attack}
In our last set of experiments, we investigate \mymethod's resistance to two common attacks, which are defined below.

\vspace{3pt}
\noindent
\textbf{Re-identification Attack.}
Suppose the attackers are able to obtain some spatial locations of a user, \eg users can be easily tracked when they enter some public sensitive zones like train stations and shopping malls. With this subtrajectory $T_{a}$ \wrt a specific user $u$ as external knowledge, the attackers aim to identify the individual $u$ in the published dataset $\Set{T}_{syn}$, and acquire information about the attacked user $u$'s whole trajectory. The defense goal is to ensure that there are at least $\kappa$ traces in $\Set{T}_{syn}$ having their similarity distances to $T_a$ below a given threshold $\vartheta$, so that the attacker is unable to identify the true user $u$ even with $u$'s subtrajectory $T_{a}$. Formally, the re-identification attack is defined below:
\begin{definition}[Re-Identification Attack]
Let $M_{T_a}$ denote the set of trajectories in $\Set{T}_{syn}$ that are similar to a given subtrajectory $T_{a}$, $sim(T, T')$ measures the similarity between two trajectories $T$ and $T'$, and $Z$ denote the selected sensitive zones.
{\setlength{\abovedisplayskip}{3pt}
\setlength{\belowdisplayskip}{3pt}
\begin{equation}
    M_{T_a} = \{T|T\in \Set{T}_{syn}, sim(T_a, T\cap Z)\leq \vartheta\}.
\end{equation}}If $|M_a|>\kappa$, the attack is successfully defensed.
\end{definition}

In our experiments, we set $2\times 2$ grids
in the central area of the map as the sensitive zone, use DTW distance as the similarity distance, and set $\vartheta=0.2\times sim_{max}$, where $sim_{max}$ denotes the maximum DTW distance of real trajectories in sensitive zones.

\vspace{3pt}
\noindent
\textbf{Outlier attack.}
In this scenario, the attackers focus on some records that have irregular attributes like outliers, which are relatively far from their neighbors, \eg trajectories with extremely long travel distance or unusual start/end locations that others seldom reach. Hence, it is easier to distinguish them from other normal trajectories, which may cause privacy leakage. We follow~\cite{ccs18adatrace}, and define this attack as follows:
\begin{definition}[Outlier Attack]
    For any outlier $T_o\in \Set{T}_{syn}$, let $M_{T_o}$ denote the matching set that includes all the real trajectories in $\Set{T}$ that are similar to $T_o$, and similarity function $sim(\cdot)$ measures the travel distance difference between two trajectories with $\delta=0.25\times sim_{max}$. If $|M_{T_o}|>\kappa$, the attack is successfully defensed. 
    {\setlength{\abovedisplayskip}{3pt}
    \setlength{\belowdisplayskip}{3pt}
    \begin{equation}
        M_{T_o} = \{T|T\in \Set{T}, sim(T_o, T)\leq \delta\}
    \end{equation}}
\end{definition}
To measure our model's defense ability to the aforementioned attacks, we define\textit{resilience ratio} as follows:
{\setlength{\abovedisplayskip}{3pt}
\setlength{\belowdisplayskip}{3pt}
\begin{equation}
    Resilience\ Ratio=\frac{\sum_{T_a\in D_a}{\mathbbm{1}(|M_{T_a}|>\kappa)}}{|D_a|}
\end{equation}}where $D_a$ is the set of trajectories to be attacked.

The results on Oldenburg dataset are shown in Figure~\ref{fig:attack}, while the results of other datasets are omitted due to similar trend and limited space. The results indicate that \mymethod has impressive ability to resist these two attacks: more than 88\% trajectories can be successfully protected in re-identification attack, and all trajectories are hidden from outliers when the privacy parameter $\kappa$ changes its value from 2 to 10. On the contrary, \baseline cannot provide provable protections to these attacks, especially when the demand of protection is more strict (large privacy parameter $\kappa$). We contribute the superiority of \mymethod to the synthesis design: since the published trajectories are synthesized from learned patterns, they do not resemble any real trajectory. Thus, it is much more difficult for attackers to identify traces that they are interested in.

\section{Conclusions}\label{sec:conclusion}
In this paper, we develop a neat yet effective trajectory synthesis framework under the rigorous privacy of LDP, called \mymethod, which achieves strong utility and efficiency simultaneously. Besides, \mymethod can provide deterministic resilience against common location-based attacks. We also provide a theoretical guideline for selecting the grid granularity without consuming any privacy budgets. Extensive experiments conducted on three datasets demonstrate the superiority of \mymethod. In the near future, we aim to extract more complex patterns from user's trajectory (like second-order Markov chain and average speed) to further enhance the authenticity of synthetic trajectories, 
and to investigate the LDP-based synthesis problem on streaming trajectories to empower real-time location-based applications. 
\appendix
\section{appendix}\label{appendix-proofs}


\renewcommand\thefigure{\Alph{section}\arabic{figure}}
\renewcommand\thetable{\Alph{section}\arabic{table}}    
\setcounter{figure}{0}
\setcounter{table}{0}

\begin{theorem}[Mean and Variance of an OUE Ratio]
\label{theorem-ratio}
Given the unbiased frequency estimations of value $x$ and value $y$ (\ie $\tilde g(x)$ and $\tilde g(y)$) with OUE, Equation~\eqref{equ:mean-f}, and Equation~\eqref{equ:var-gxgy} define the approximated mean and variance, respectively.
{\setlength{\abovedisplayskip}{3pt}
\setlength{\belowdisplayskip}{3pt}
\begin{equation}\label{equ:mean-f}
{\mean}^*[\frac{\tilde g(x)}{\tilde g(y)}] = \frac{f_x}{f_y},
\end{equation}}
{\setlength{\abovedisplayskip}{3pt}
\setlength{\belowdisplayskip}{3pt}
\begin{equation}\label{equ:var-gxgy}
{\var}^*[\frac{\tilde g(x)}{\tilde g(y)}] =
\frac{\left(f_x\right)^2}{\left(f_y\right)^2}\left[\frac{\sigma_x^2}{\left(f_x\right)^2}-2 \frac{\operatorname{Cov}(x, y)}{f_x f_y}+\frac{\sigma_y^2}{\left(f_y\right)^2}\right],
\end{equation}}where $f_x$ and $f_y$ are the true frequencies of $x$ and $y$, respectively; $\sigma_x^2$ and $\sigma_y^2$ are the variance of OUE for estimators $\tilde g(x)$ and $\tilde g(y)$, respectively.
\end{theorem}

\begin{proof}[Proof of Theorem~\ref{theorem-ratio}]
For any function $f(X,Y)$, we can choose the expansion point to be $\theta = (\mu_x, \mu_y)$, and the first order Taylor series approximation for $f(X,Y)$ is:
{\setlength{\abovedisplayskip}{3pt}
\setlength{\belowdisplayskip}{3pt}
\begin{equation*}
\begin{aligned}
\mean[f(X, Y)]
& \approx \mean[f(\theta)]+\mean\left[f_x^{\prime}({\theta})\left(X-\mu_x\right)\right]+\mean\left[f_y^{\prime}({\theta})\left(Y-\mu_y\right)\right] \\
&=\mean[f(\theta)]+f_x^{\prime}({\theta}) \mean\left[\left(X-\mu_x\right)\right]+f_y^{\prime}({\theta}) \mean\left[\left(Y-\mu_y\right)\right] \\
&=f\left(\mu_x, \mu_y\right).
\end{aligned}
\end{equation*}}Let $f(x,y) = x/y$, and the approximation holds ${\mean}^*[f(X,Y)] = f(\mu_x, \mu_y) = \mu_x / \mu_y$.
Therefore, the mean of an OUE ratio $\tilde g(x)/\tilde g(y)$ approximates $f_x/f_y$, where $\mean[\tilde g(x)] = f_x$ and $\mean[\tilde g(y)] = f_y$. 

Besides, the variance of $f(X,Y)$ is:
{\setlength{\abovedisplayskip}{3pt}
\setlength{\belowdisplayskip}{3pt}
\begin{equation*}
\begin{aligned}
{\var}[f(X, Y)]=\var\left\{[f(X, Y)-E(f(X, Y))]^2\right\}
 \approx \var\left\{[f(X, Y)-f(\theta)]^2\right\}
\end{aligned}
\end{equation*}}Then using the first order Taylor expansion for $f(X, Y)$ around $\theta$:
{\setlength{\abovedisplayskip}{3pt}
\setlength{\belowdisplayskip}{3pt}
\begin{equation*}
\begin{aligned}
\var[f(X, Y)] & \approx \mean\left\{\left[f(\theta)+f_x^{\prime}({\theta})\left(X-\theta_x\right)+f_y^{\prime}(\theta)\left(Y-\theta_y\right)-f(\theta)\right]^2\right\} \\
&=f_x^{\prime 2}({\theta}) \var(X)+2 f_x^{\prime}({\theta}) f_y^{\prime}(\theta) \operatorname{Cov}(X, Y)+f_y^{\prime 2}({\theta}) \var(Y)
\end{aligned}
\end{equation*}}where $f(x,y) = x/y$, and the approximated variance is:
{\setlength{\abovedisplayskip}{3pt}
\setlength{\belowdisplayskip}{3pt}
\begin{equation*}
\begin{aligned}
{\var}^* [X / Y]
&=\frac{\left(\mu_x\right)^2}{\left(\mu_y\right)^2}\left[\frac{\sigma_x^2}{\left(\mu_x\right)^2}-2 \frac{\operatorname{Cov}(X, Y)}{\mu_x \mu_y}+\frac{\sigma_y^2}{\left(\mu_y\right)^2}\right]
\end{aligned}
\end{equation*}}
\end{proof}


\begin{acks}
This work was supported by the NSFC under Grants No. (62025206, 61972338, and 62102351). Lu Chen is the corresponding author of the work.
\end{acks}


\bibliographystyle{ACM-Reference-Format}
\balance

\bibliography{ref}
\balance

\end{document}